%% file: ms.tex
\documentclass[acmsmall]{acmart}

\usepackage{booktabs} % For formal tables
\usepackage{standalone}
\usepackage[T1]{fontenc}

\usepackage{tikzsymbols}
\usepackage{listings}
\usepackage{pbox}
\usepackage{xcolor,colortbl}
\usepackage{wrapfig}
\usepackage{framed}
\usepackage{microtype}

\usepackage{multirow}
\usepackage[ruled]{algorithm2e} % For algorithms
\usepackage{enumitem}
\usepackage{subcaption}
\usepackage{graphicx}
\usepackage{tikz}
\usetikzlibrary{matrix, positioning}

\SetAlFnt{\small}
\SetAlCapFnt{\small}
\SetAlCapNameFnt{\small}
\SetAlCapHSkip{0pt}
\IncMargin{-\parindent}
\acmJournal{PACMPL}
\acmVolume{1}
\acmNumber{OOPSLA}
\acmArticle{1}
\acmYear{2018}
\acmMonth{1}
\acmDOI{} % \acmDOI{10.1145/nnnnnnn.nnnnnnn}
\startPage{1}

\newtheorem*{property*}{Property}
\newtheorem*{lemma*}{Lemma}
\newtheorem*{precond*}{Precondition}

\setcopyright{none}
\begin{document}

\setlist[description]{font=\normalfont\space}

% Title portion
\title{Faster Variational Execution with Transparent Bytecode Transformation}
\author{Chu-Pan Wong}
\affiliation{%
	\institution{Carnegie Mellon University}
	\country{USA}
}
% \orcid{1234-5678-9012-3456}
\author{Jens Meinicke}
\affiliation{%
	\institution{Carnegie Mellon University}
	\country{USA}
}
\affiliation{%
	\institution{University of Magdeburg}
	\country{Germany}
}
\author{Lukas Lazarek}
\affiliation{%
	\institution{Northwestern University}
	\country{USA}
}
\author{Christian K\"astner}
\affiliation{%
	\institution{Carnegie Mellon University}
	\country{USA}
}

\input{macros}

\begin{abstract}

Variational execution is a novel dynamic analysis technique for exploring
highly configurable systems and accurately tracking information flow. It is
able to efficiently analyze many configurations by aggressively sharing
redundancies of program executions. The idea of variational execution has been
demonstrated to be effective in exploring variations in the program, especially
when the configuration space grows out of control. Existing implementations of
variational execution often require heavy lifting of the runtime interpreter,
which is painstaking and error-prone. Furthermore, the performance of this
approach is suboptimal. For example, the state-of-the-art variational execution
interpreter for Java, VarexJ, slows down executions by 100 to 800~times over a
single execution for small to medium size Java programs. Instead of modifying
existing JVMs, we propose to transform existing bytecode to make it
variational, so it can be executed on an unmodified commodity JVM.  Our
evaluation shows a dramatic improvement on performance over the
state-of-the-art, with a speedup of 2 to 46 times, and high efficiency in
sharing computations. 

\end{abstract}


%
% End generated code
%

% We no longer use \terms command
% \terms{Configurable systems, Variational execution, Performance}


% \thanks{This work is supported by the National Science Foundation,
%   under grant CNS-0435060, grant CCR-0325197 and grant EN-CS-0329609.

%   Author's addresses: G. Zhou, Computer Science Department, College of
%   William and Mary; Y. Wu {and} J. A. Stankovic, Computer Science
%   Department, University of Virginia; T. Yan, Eaton Innovation Center;
%   T. He, Computer Science Department, University of Minnesota; C.
%   Huang, Google; T. F. Abdelzaher, (Current address) NASA Ames
%   Research Center, Moffett Field, California 94035.}

\maketitle

% The default list of authors is too long for headers}
% \renewcommand{\shortauthors}{G. Zhou et al.}

% \input{samplebody-journals}

\section{Introduction}
\label{sec:intro}

Computer programs often come with variations that allow programs to act
differently according to a user's need. A classic example of variation is
command-line options or configuration files, which often trigger new behaviors
or tweak existing functionalities. Configuration options are widely used
because they provide flexible extension points for adding new features.
However, this flexibility often comes at the cost of potential feature
conflicts. Feature conflicts arise when one feature interferes with another in
an unintended way (also known as feature interaction
problem~\citep{Calder:2003hra, Nhlabatsi:2008wq}). Problems similar to feature
interactions have also been studied in other contexts, such as security
policies, where the sensitivity of a program to various private values is explored
by comparing different executions varying in privacy levels for values~\cite{AYF+:PLAS13, AF:POPL12}. To tackle these problems, we need
a principled yet efficient way of detecting and managing \emph{interactions of
variations}.

There are several challenges posed by \emph{interactions of variations}.  First,
the number of possible interactions is exponential to the number of variations,
making exhaustive testing of all configurations unrealistic in most practical
cases. Second, interactions are often impossible to foresee, for example in the
context of plugin-based systems where plugins are often developed
independently by different developers, making it easy to miss interactions
when sampling only few configurations~\cite{MKR+2015, NL:CSUR11, TAK+:CSUR14}. Third, effects of variations often
propagate globally in programs, making it hard to detect interactions involving
lots of variations even with complex data-flow analysis~\cite{LKB:ASE14}.

To tackle these challenges, researchers have proposed dynamic analysis
techniques that analyze the effects of multiple variations by efficiently
tracking variations at runtime. Researchers have applied these techniques to
various scenarios, such as testing highly configurable
systems~\cite{NKN:ICSE14}, understanding feature
interactions~\cite{MWK+:ASE16} and configuration faults~\cite{SAF:SOSP07}, monitoring information flow of sensitive
data~\cite{AF:POPL12, AYF+:PLAS13, DP:SP10, KLZ+:SP12, KKS+:ASPLOS16, KKS+:ASPLOS15}, and detecting inconsistent updates~\cite{HC:ICSE13, MB:USENIX12,TXZ09}. These techniques are similar, and often
called differently in different communities, such as variability-aware
execution~\cite{Kastner:2012jna, NKN:ICSE14, MWK+:ASE16}, faceted
execution~\cite{AF:POPL12, AYF+:PLAS13}, coalescing
execution~\cite{SBZ+:ICES11}, shared execution~\cite{KKB:ISSRE12}, and
multi-execution~\cite{DeGroef:2012jc,DP:SP10}. Our work is built on
these ideas, and we use the name variational execution in this work, as
we target primarily analyzing and testing configuration options in programs.

Previous studies have shown that variational execution can be useful in many 
scenarios with promising results. 
Since variational execution itself is not a core contribution in this work, we
defer the comprehensive discussion of different applications to Section~\ref{related}.
Nonetheless, we highlight a few interesting results:
\citet{NKN:ICSE14} applied variational execution to identify plugin conflicts
in WordPress. Their variational execution engine can analyze
$2^{50}$ combinations out of $50$ plugins within seven minutes and found a
previously unknown plugin conflict. \citet{MWK+:ASE16} used
variational execution to understand the shape of configuration spaces in
different programs and found interesting characteristics of how options
interact. \citet{AF:POPL12} and \citet{AYF+:PLAS13} demonstrated usefulness of variational
execution in guaranteeing non-interference of sensitive data between different
confidentiality levels. Their prototype implementations in JavaScript and Jeeves
can prevent cross-site scripting attacks and handle complex information flow in
a conference management system.
\citet{SBZ+:ICES11} showed that a form of variational execution can exploit similarities
among data processing of similar inputs and gain a speedup of
2.3 without precision lost.
Although variational execution has been explored before with promising results,
a universal and scalable implementation is still missing. Existing
implementations typically have severe scalability issues and work only with 
small academic examples. With a better implementation, we have a better chance
of scaling existing applications and applying variational
execution to broader application scenarios and more use cases.

Existing implementations rely on either
\emph{manual modification to the source code}~\cite{Schmitz:2016be, AYF+:PLAS13,
schmitzFacetedSecure2018} or
\emph{modification to the language interpreter}~\cite{MWK+:ASE16, NKN:ICSE14}: On the one hand,
variational execution can be achieved by writing the source code to use some
libraries or programming language constructs, so that the programs compute
with multiple values in parallel~\cite{Schmitz:2016be, AYF+:PLAS13,
schmitzFacetedSecure2018}.
Implementations of this kind put a heavy burden on developers because the use of
these libraries or language constructs usually obscures the original programs.
Moreover, rewriting existing programs is often tedious and error-prone. On the other hand,
variational execution can be achieved by executing a normal program with a
special execution engine, such as an interpreter that tracks multiple
values in parallel with special
operational semantics for each instruction~\cite{MWK+:ASE16, NKN:ICSE14,
AF:POPL12}.  Modified interpreters often suffer from a conflict between
functionalities and engineering effort: It would be painstaking to modify a
mature interpreter like OpenJDK, though it fully supports all functionalities of
the language, whereas it takes less engineering effort to modify a
research interpreter such as Java PathFinder~\cite{Havelund:2000ed}, which
however provides incomplete language support and often mediocre performance.
\looseness=-1

We present a new way of implementing variational execution. Our
approach sidesteps \emph{manual modification to the source code} and
\emph{brittle modification to the language interpreter}. The key idea
is to automatically transform programs in their \emph{intermediate
representation}.  Specifically, we transparently modify Java bytecode automatically
to mirror the effects of a manual rewrite. The resulting bytecode
can then be executed on an unmodified commodity JVM.

\looseness=-1
Transforming programs at the intermediate language level has several benefits.
First, intermediate languages often have simple forms and strong specifications,
both of which facilitate automatic transformation. Second, source code is not
required, allowing us to transform also
libraries used in the target programs. We can even analyze other programming languages
that are compilable to the same intermediate language.
Third, existing optimizations of the execution engine
can be reused;  in our case, our transformed bytecode can
take advantage of just-in-time compilation and other optimizations provided by modern JVMs.
Finally, modifications at the intermediate level
remain portable. Our transformed bytecode can be executed on any JVM that
implements the JVM specification.

Transformations are nontrivial and not always local.
While many bytecode instructions can be transformed in isolation,
encoding conditional control flow in a commodity JVM requires careful encoding, such
that both branches of control-flow decisions can be executed in
different configurations, before subsequent computations are merged
again, to maximize sharing overall. In additional, data-flow
analyses are required to handle values on the operand stack
between blocks and object initialization sequences for variational
execution. Finally, we perform additional optimizations to statically 
pinpoint instructions that do not need to be transformed, 
because they are guaranteed to be not related
to variations in the program.

We formally prove that our transformation of 
control flow is correct, statically guarantee optimal sharing 
for a large subset of possible control-flow graphs.
Additionally, we empirically evaluate performance, 
comparing execution time and memory consumption on seven highly
configurable systems against VarexJ, a state-of-the-art variational execution
implementation. The results show that our approach is 2 to 46~times faster
than VarexJ, with 75~percent less memory. The performance results
also indicate that our approach is efficient for analyzing highly configurable
systems in practice.  
\looseness=-1
% We further perform a runtime monitoring of method
% invocations that we cannot guarantee optimal sharing. The monitoring results
% confirm that we achieve $99.8\%$ of optimal sharing on our subject programs,
% not to mention those methods that we statically guarantee optimality.

We summarize our contributions as follow:
\begin{itemize}

	\item We propose a novel strategy for variational execution using automatic
		bytecode transformation,
		without any manual modifications to the source code or to the language
		interpreter.

	\item We prove that our automatic transformation of bytecode is correct for
		all control-flow graphs and optimal with regard to sharing for a large subset.

	\item We propose further optimizations by performing data-flow analysis and using specialized data structures.

	\item We implement a bytecode transformation tool that covers nearly the entire
		instruction set of the Java language, with minor exceptions that we
		explain in Section~\ref{sec:optim}. The transformed bytecode is portable to any 
	implementation of the JVM specification.

	\item An empirical evaluation with 7~subject systems showing that our
		approach is up to 46~times faster while saving up to 75~percent memory when compared to the state-of-the-art. In addition to statically
		guaranteeing optimal sharing for $89.7$~percent of methods, our approach
		achieves optimal sharing at runtime for $99.8$~percent of all other method executions.

\end{itemize}

We hope that the way we transform bytecode can inspire more efficient
implementation of similar techniques such as symbolic execution. Although we
focus on Java bytecode in this work, we can potentially generalize the core
ideas to other programming languages and other analyses, by performing a similar transformation at
the well-defined intermediate representation form of existing compiler
frameworks like LLVM. 

\section{Background and Motivation}
\label{sec:motivation}

As necessary background for our approach, we introduce the core concepts of
variational execution and show how it can be achieved
with manual source-code transformation, hinting at the key ideas of our automatic
bytecode transformation.

\subsection{Variational Execution}

\begin{figure*}[tp]
	\centering
	\includestandalone[print,sort,mode=buildmissing,width=\linewidth]{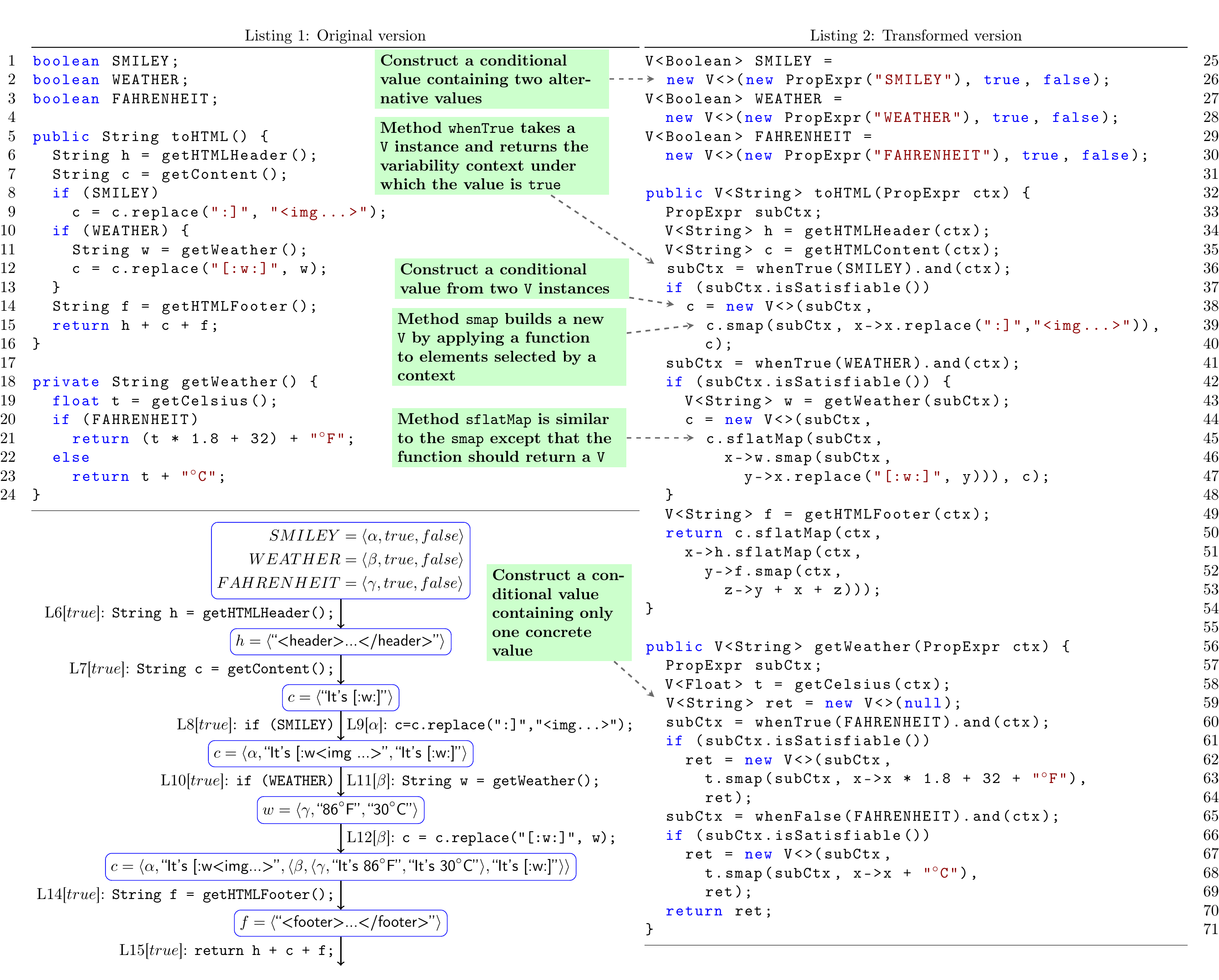}
	\caption{\footnotesize Running example of this paper, modeled after WordPress~\cite{MWK+:ASE16}. Listing 1
	shows the original source code without variational execution. Bottom left
illustrates variational execution by showing the execution trace. Listing 2
hints at our variational execution transformation. The
transformation is shown in Java for better readability.}
	\label{fig:running}
	\vspace*{-1em}
\end{figure*}

There are two main concepts that distinguish variational execution from concrete
execution: \textbf{conditional values} and \textbf{variability contexts}.

The key idea of variational execution is to execute a program with
\emph{concrete values}, but support \emph{multiple alternative concrete values}
for different configurations. That is, whereas each variable has one
concrete value in concrete execution (e.g., $x = 1$), the concrete value of a variable
may depend on the configuration in variational execution---we say the variable
has a \textbf{conditional value}~\cite{EW:GTTSE13}. A conditional value
does not store a separate value for each configuration (exponentially many),
but partitions the configuration space into \emph{partial spaces} which share
the same value. That is, all configurations sharing the same concrete value are
represented only once in the conditional value. Partial configuration spaces
are expressed through propositional formulas over configuration options, such
as $(a \lor b) \land \neg c$ representing the potentially large set of all
configurations in which configuration options $a$ or $b$ are selected but not
$c$; a tautology (denoted as \texttt{true}) describes all configurations, a
contradiction (denoted as \texttt{false}) none.  Conditional values are
typically expressed through possibly-nested choices over formulas (or
if-then-else expressions), such as $x = \langle a, \langle \neg b \lor c, 1, 3
\rangle, 2 \rangle$, which means: $x$ has the value $1$ in the partial space $a
\land (\neg b \lor c)$, $3$ in $a \land \neg (\neg b \lor c)$, and $2$ in $\neg
a$. With this representation, we can reason about configuration spaces with SAT
solvers and BDDs.
\looseness=-1

Variational execution uses conditional values with the notion of performing a
computation conditionally in a \textbf{variability context}, similar to
a path condition in symbolic execution: An operation will
only modify values in the part of the configuration space indicated by the current variability context (that is, we
conceptually split the execution). Again, formulas over configuration options
are used to express the variability context.

We assume a finite configuration space in which we know concrete values for all
configuration options.
Throughout this paper, we focus on boolean options when we discuss configuration
options because they are common and
easy to reason about with standard tools. Other options (e.g., strings, numeric
values) with \emph{finite domains} can be encoded as a set of boolean options.
Using other solvers to circumvent the encoding is possible as well. The support
of options with finite domains is entirely hidden behind the abstraction of
conditional values, so it is orthogonal to the discussion in this paper.
\looseness=-1

Operations on conditional values can often be shared.
If none of the used variables have alternative values, an
instruction only needs to be executed once for all configurations (we say that
we are executing under the \texttt{true} context). We begin execution in the
\texttt{true} context, and only split into restricted variability contexts when
configuration options influence execution---directly or indirectly. This
conservative execution splitting strategy allows us to aggressively share
executions that would otherwise be repeated once per configuration. This sharing
avoids nonessential computations and makes variational execution efficient in
many scenarios.

\paragraph{Comparing to Symbolic Execution} Despite some similar concepts,
there are important differences between variational execution and symbolic
execution.  A \emph{conditional value} in variational execution is
fundamentally different from a \emph{symbolic value} in symbolic execution, in
that the former represents \emph{a finite number of concrete values} while the
latter often represents \emph{an infinite set of possible values} of a given
data type.  Unlike symbolic execution where operations are carried out on
symbolic values, variational execution always computes with concrete values;
symbols are used only to describe configuration spaces for
distinguishing alternatives and for describing contexts, but never intermix
with concrete values. For this reason, loop bounds are always known concrete
values in variational execution, and we avoid other undecidability problems.
By considering finite configuration spaces, reasoning about configuration space
of conditional values involves \emph{inexpensive and decidable} satisfiability
checks with SAT solvers or BDDs, while symbolic execution is often limited by
\emph{expensive} constraint solving and the types of theories the underlying
constraint solver supports. For instance, reasoning about array elements in
variational execution is fast, because we know the concrete array indexes and
elements, in contrast to symbolic execution where a symbolic array index can
dramatically slow down constraint solving because it can potentially refer to
every element in the array.  \looseness=-1

Furthermore, variational execution has different concepts of managing state and
forking and joining when compared to symbolic execution. Symbolic execution
often forks new states either completely or partially at every conditional
branch, often resulting into exponentially many paths in practice, commonly
known as the path explosion problem.  For example, \citet{MWK+:ASE16} has
demonstrated that state-of-the-art symbolic execution implementations for Java
split of separate executions on variability and share only a common prefix.
Some symbolic execution engines merge state from different paths to share
executions after control flow decisions, for example, introducing new symbolic
values or using \emph{if-then-else} expressions to represent differences among
values from different paths---different designs make different tradeoffs with
regard to performance, precision, and implementation
effort~\cite{Baldoni:2018kg,SNG+:FSE}.  Our implementation of variational
execution uses a design that maximizes sharing. It maintains a single
representation of all state throughout the execution where differences are
represented at fine granularity (variables and fields) with conditional values.
State is always modified under the current variability context, which is
equivalent to merging state after every single statement. In addition, we join
control flow as early as possible to avoid repeating executions, as we will
discuss in Section~\ref{sec:control}.

\paragraph{Example of Variational Execution}
As an example, consider Listing 1 in Figure~\ref{fig:running}, a simplified
implementation of a blogging system modeled after WordPress. The blogging system
has three variations, based on options for smiley rendering and inlining weather reports,
which affect how HTML code is generated.
In its current form, there is an issue: if both \texttt{SMILEY}
and \texttt{WEATHER} are enabled, the replacement of a smiley image takes
precedence and breaks the expansion of weather information, resulting in
outputs like ``[:w\Laughey''.
\looseness=-1

In order to ensure the absence of interaction bugs like
this, typical testing techniques would try all combinations one by one,
resulting into 8 executions of the same program in this case. Moreover, single
executions alone reveal little information about the causes of interaction bugs,
especially for cases where interactions of options have global effects on the
execution.

Variational execution is much more efficient for detecting interactions like
this.
The execution trace in the bottom left of Figure~\ref{fig:running} illustrates
how variational execution explores all possible interactions among
\texttt{SMILEY}, \texttt{WEATHER} and \texttt{FAHRENHEIT} in \emph{a single
run}. Boxes represent relevant program states and arrows denote execution steps.
The executed statements are displayed beside arrows, together with the
variability contexts. An execution trace like this can also be generated by
logging and aligning concrete executions of all possible configurations, but
\citet{2018arXiv} showed that variational execution is much more efficient,
sidestepping correctness and performance issues of alignment.

After marking the three \texttt{boolean} fields as options (e.g. via Java
annotation), variational execution initializes them with 
\emph{conditional values}, representing both \texttt{true} and \texttt{false}.
The symbols $\alpha$, $\beta$, $\gamma$ denote the three options respectively.
Variational execution runs Line 6 and Line 7 once under the \emph{variability
context} of \texttt{true}, meaning that they are shared across all configurations.
Sharing like this enables variational execution to explore large configuration
spaces efficiently. To highlight sharing, we put all shared statements 
to the left of the arrows in the execution trace.
The execution is split when it comes to the
first \texttt{if} statement, where \texttt{c} is modified only under the
\emph{variability context} of \texttt{SMILEY}. At this point, the content of
\texttt{c} changes from \emph{containing one value for all configurations} to
\emph{having two alternative values depending on option} \texttt{SMILEY}, and this
change is reflected in the \emph{conditional value} assigned to \texttt{c}.
Finally, variational execution is able to share the execution of common code
again at Line~14, after splitting executions in two \texttt{if} branches.

This example illustrates the benefits of variational execution. We can spot the
problematic interaction of \texttt{SMILEY} and \texttt{WEATHER} by inspecting 
the conditional value of \texttt{c}, as shown in the execution trace. 
In fact, all possible interactions are recorded and
detectable by inspecting \emph{conditional values} during the variational
execution. All information about how options interact can be obtained after one
single run of variational execution, in contrast to exponentially many with
normal execution, and the difference
would still not be obvious without aligning all traces of normal execution. The
effectiveness of
variational execution comes from using \emph{variability context} to manage
splitting and sharing of executions.

\paragraph{VarexJ}
The state-of-the-art implementation of variational execution for Java is
VarexJ~\cite{MWK+:ASE16}. VarexJ is implemented on top of Java
PathFinder's (JPF) interpreter for bytecode~\cite{Havelund:2000ed}.  For this
reason, VarexJ also inherits several limitations that restrict the programs it
can analyze.  First, JPF is not a complete implementation of the JVM
specification, so it provides incomplete support for language features, such as
concurrency and native methods. Second, unlike commodity JVMs, JPF does not
provide any optimizations over programs being executed, such as just-in-time compilation. Third,
executing programs on JPF is slow because JPF itself is implemented as a
single-process Java application. The nature of meta-circular interpreting
causes a significant performance penalty. Our approach sidesteps these problems by
not modifying the JVM, but transforming the Java bytecode.

\subsection{A Manual Rewrite}

To illustrate how variational execution can be achieved on a commodity JVM,
we illustrate how the source code of our WordPress example in
Figure~\ref{fig:running} can be manually rewritten
in Listing~2 of Figure~\ref{fig:running}. We show the rewrite in Java source
code for better readability, as the same program in bytecode is typically longer
and harder to read, potentially obscuring the essential ideas of our rewriting.
This manual rewrite in Listing~2 also introduces the key ideas (highlighted as
floating boxes) used later in our automated bytecode transformation. A rewrite
in bytecode is also available in the appendix.

We introduce \emph{variability contexts} in all methods, represented by instances of
the \texttt{PropExpr} class, which model propositional expression over
configuration options.
Variables are rewritten to use a new \texttt{V} type to store \emph{conditional values},
either a single value for all configurations or different values for different
configurations.
To manipulate values in \texttt{V} objects, we use \texttt{smap} and \texttt{sflatMap}
methods.  
The \texttt{smap} method applies a function to each alternative value of a
\texttt{V}, and the \texttt{sflatMap} method does the same but allows to split
configuration spaces, producing more alternatives.
For example, the operation \texttt{v.smap(ctx, f)} on a conditional value
\texttt{v} of type \texttt{V<T>} takes as arguments (1) a variability context \texttt{ctx}
and (2) a function literal \texttt{f} of type \texttt{T => U}, representing the
pending operation. It returns a new V instance of type \texttt{V<U>} that results from applying the
function \texttt{f} to each concrete element that exists under \texttt{ctx} in \texttt{v} (recall that
a conditional value stores concrete values along with the variability contexts
under which they exist). The \texttt{sflatMap} method works similarly,
but takes functions of type \texttt{T => V<U>}.

Note that the manual rewrite shown in Listing~2 is not exactly the same as our 
bytecode transformation, but close enough to show the key ideas. An automatic
rewrite of our running example in bytecode is available in the appendix for reference.
We will discuss in detail how such a rewrite can be automated in bytecode in
Section~\ref{sec:trans} and Section~\ref{sec:control}. 
Nonetheless, we can already preview a few key points from this manual rewrite:

\begin{itemize}
	\item Variables store conditional values, represented by \texttt{V} objects.
	\item Most operations on conditional values (e.g., calling the
		\texttt{replace} method, String concatenation) are redirected with
		\texttt{smap} and \texttt{sflatMap} and applied to all alternative
		concrete values. It fact, this replacement is sufficient for most
		bytecode instructions, as we will see in Section~\ref{subsec:basic}.
	\item Both the \texttt{if} branch and the \texttt{else} branch of an
		\texttt{if-else} statement are transformed into an \texttt{if} statement,
		a statement that checks whether there exists any partial configuration under
		which the surrounded code will be executed. If such a partial
		configuration exists, the surrounded code will be executed under a
		restricted variability context (e.g., Line 36--40). We will discuss transformation of
		control-transfer instructions in Section~\ref{sec:control}.
	\item All method calls have one additional parameter \texttt{ctx},
		representing the variability context under which this method is called.
		The variability context restricts all instructions of that method
		invocation. Also, multiple \texttt{return} statements in the same method
		are replaced with temporary assignment to a local variable, which is
		returned in the end of the method. Transformation of method calls and
		method returns will be further discussed in Section~\ref{subsec:method}.
\end{itemize}

The transformation from normal code to variational code is nontrivial and
obscures the program. For example, we almost double the size of Listing~1 in
order to transform a simple example into a variational execution version.
The introduction of \texttt{smap} calls and complicated control-transfer structures
also obscure the intention of the original program, making it
hard to understand and debug. This puts a
heavy burden on the developers to understand variational execution
and how to use it correctly. All of these issues can be resolved if we
adapt an automatic transformation approach that is transparent to
developers. As we will see later in Section~\ref{sec:optim}, our transformation
is also able to automatically decide which parts of a program need to
be transformed, as it is likely that some parts are not related to variations,
such as the code before the first \texttt{if} statement (Line~7) and the code
after the second \texttt{if} statement (Line~14) in Listing~1.

\section{Bytecode Transformation}
\label{sec:trans}

% \{ck: move this paragraph to an implementation section:} We use the ASM\footnote{\url{http://asm.ow2.org/}} library to implement our data
% flow analysis and transformation of bytecode. ASM is widely used
% for dynamic analysis. \{citations?} Our implementation is available on
% GitHub.\footnote{The link is hidden for double-blind reviewing} The choice of
% bytecode engineering library is orthogonal to our bytecode transformation
% approach.

We discuss our transformation in two steps. First, in this section, we
discuss how to transform all instructions that are executed in
a given variability context. The transformation of control flow, which may change
variability context, is nontrivial and
orthogonal, so we discuss it second in Section~\ref{sec:control}.
We describe transformations for similar instructions together, following
the grouping of the JVM
specification~\cite{Buckley:2015tu}.
\looseness=-1

In a nutshell, we transform each bytecode instruction of the original program into
a sequence of bytecode instructions. Ideally, the transformation of most
instructions should be local, meaning that the transformation of the current
instruction should not be affected by other instructions around it. However,
this locality assumption is not generally possible because an instruction often affects another
instruction by leaving data on the operand stack. The operand stack is used internally in
the JVM for exchanging data between instructions. Some instructions load values
(e.g., constants or values from
local variables or fields) onto the operand stack, while other instructions take
values from the operand stack and operate on them. 
Results might be pushed back onto the operand stack as a result of an operation.
The operand stack is also used to prepare parameters to be
passed to method invocations and to receive return values.

To assist local transformation of individual instructions, we introduce several
transformation invariants:
\looseness=-1

\begin{enumerate}[label=\textit{Invariant \arabic*}, align=left]

	\item \label{i:local} All local variables and fields store conditional values.

	\item \label{i:stack} All values on the operand stack are conditional values.

	\item \label{i:method} All methods take conditional
		values as parameters and return conditional values.

\end{enumerate}

We ensure that these invariants hold \emph{before and after} the execution of each
transformed bytecode sequence. They help us establish a common ground about the
environment, enabling concise transformations of most instructions.
In addition, we assume that each instruction is executed in a local variability context.
We will explain how variability contexts are propagated and changed as part
of our discussion of control flow in Section~\ref{sec:control}.

% \{ck: this paragraph does not seem relevant here:}
% As mentioned, \emph{conditional value} and \emph{variability context} are two
% important concepts of variational execution. While our invariants facilitate
% using conditional values, we allocate additional local variables to store
% variability contexts. In variational execution, methods are always called under
% a variability context, specified as an additional parameter to the method call
% (e.g., Line~32 of Listing~2). The \emph{method variability context} restricts
% the execution of all instructions inside that method. The variability context of
% individual instructions can be further restricted by control flow decisions,
% such as Lines~38--40 in Listing~2. We will discuss transformation of control flow in
% Section~\ref{sec:control}.

% With a few exceptions, most instructions are transformed by expanding each of
% them into \texttt{smap} or \texttt{sflatMap} calls on conditional values. In the
% rest of this subsection, we present our transformation scheme for some basic
% instruction types.

\subsection{Basic Lifting}
\label{subsec:basic}

To achieve our invariants, we change all parameters and local variables
in a method frame to the \texttt{V} type to store conditional values.
Primitive types are boxed in the process.

\paragraph{Load and Store Instructions} Load and store instructions transfer
values between the local variables and the operand stack.
Since we assume local variables and stack values to represent conditional
values (\ref{i:local}, \ref{i:stack}), we can directly load them with the \instr{aload} instruction
(replacing load instructions for primitive types if needed).
Store instructions require more attention, because they may be executed under a restricted
variability context, in which case not all values shall be overwritten.
For example, suppose we have $x = 1$ under context \texttt{true}, but
store $2$ to $x$ under context \texttt{A}, then $x$ stores the conditional value
 $\langle A, 2, 1\rangle$ instead of 2. To this end, we
always create a new conditional value, compressing the updated values under the current context with
possibly unaffected old values. As an example, consider the \texttt{V}
constructor call when \texttt{c} is updated in Line 38--40 of Listing 2.
\looseness=-1

\paragraph{Arithmetic and Type Conversion Instructions} Arithmetic and type
conversion instructions compute a result based on one or two values
from the operand stack, and then push the result back on the operand stack.
For example, the \instr{iadd} instruction takes two $int$ values from the stack, adds
them together and pushes the result back.
Given \ref{i:stack}, we need to pop and push conditional values.
We achieve this by invoking \texttt{smap} with the current variability context on the stack's conditional values,
performing the original arithmetic or type conversion operation on each alternative concrete value.
For operations on two conditional values, we combine \texttt{sflatMap} and \texttt{smap}
to compute results for all possible combinations.
For example, the original floating point calculation in Line 21 of
Figure~\ref{fig:running} is transformed to a \texttt{smap} call in Line 63.

\paragraph{Operand Stack Management Instructions} Operand stack management
instructions directly manipulate entries on the operand stack, such as
\instr{pop} for discarding the top value, and \instr{swap} for swapping the top
two values. They work the same for conditional values and concrete values,
and therefore do not need to be transformed.
A technical subtlety in Java is that some primitive values (e.g., long, double) are represented by two
32-bit values on the stack, but only by a single reference value
for a conditional value; here we adjust stack operations accordingly.

\subsection{Method Invocation and Return}
\label{subsec:method}

Method invocations 
% (for all five instructions
% \instr{invokevirtual}, \instr{invokeinterface}, \instr{invokespecial},
% \instr{invokestatic}, and \instr{invokedynamic}) 
pass the top stack values
as arguments to the method and push the method's result back to the stack.
Non-static methods also take their receiver from the stack.
Since method arguments and return types are conditional values, just as
stack values (\ref{i:stack}, \ref{i:method}), they can be passed along
directly. If a method call has multiple receiver objects, we call the
method for each of them in the corresponding variability context and merge results
using a \texttt{sflatMap} call.
\looseness=-1

Special handling is required though in cases in which \ref{i:method}
does not hold for the target. Ideally, all classes and all methods in variational
execution should be transformed, but this is not always possible in practice because of the
environment barrier. At some point, variational programs may need to interact
with an environment that does not know about variational values and variability
contexts. The environment barrier can be at different places, depending
on how the system is implemented
(e.g., between user code and library code, between Java code and native
code, between the program and the operating system or network), but can never be avoided entirely.
When hitting the environment barrier, we have three options:
% \begin{itemize}

% For example, a program cannot write differently to the file
% system under different variability contexts, unless we implement a variational
% file system at the OS level. The environment barrier partitions all possible
% code into code that we transform and code that we do not transform.  If we
% invoke a method that we do not transform, the method being invoked cannot take
% conditional values as argument at runtime, violating our
% Invariant~\ref{i:stack}. We have two solutions to this problem.

% \item
	\paragraph{Multiple invocations} For \emph{side effect free} methods, we can invoke
the target method multiple times for each
feasible combinations of concrete argument values, merging the results
into a single conditional value.
Since the method is side effect free, invoking it repeatedly with different arguments does
not change the program states, it just forgoes potential sharing.
\looseness=-1

% \item
	\paragraph{Model classes} We can always provide variational models for the
environment, for example, replacing all reads and writes to a file with a
special implementation that can store alternative file context under different
contexts. Such model classes are common in model checking and symbolic execution~\cite{ALM:SE08,SNG+:FSE,RARF:JPF11} and have been
explored in variants of variational execution for database storage~\cite{YHA+:PLDI16}.
Model classes can also be used to provide more efficient variational implementations
for classes than would be achieved with our automated transformation, as we will discuss
in Section~\ref{sec:optim}.

% We can manually create classes that model the
% behaviors of the classes we wanted to transform, but implement it in the
% variational way, similar to Listing~2 of Figure~\ref{fig:running}. These classes
% are often called model classes, and they are common for modeling the environment
% in model checking, symbolic execution, etc. \{citation? Jean's work?} Model
% classes need not only be used to model the environment, however. We extend the
% idea to improve performance by using a version of the class which is optimized
% for variational execution. For example, we use model classes to implement
% specialized data structures that deal with conditional values more efficiently.
% Recent studies have shown that specialized variational data structures can have
% huge impact on performance~\cite{Walkingshaw:2014ch,Meng:2017bx}.  We will
% discuss model classes as an optimization in Section~\ref{sec:optim}.

	% \item
		\paragraph{Abort} Finally, we can execute the program but abort
	execution when we reach the environment barrier at runtime.
	This way, we can still support executions that do not cross the barrier,
	even though the source code refers to nonvariational methods.
	Furthermore, we can allow calls to nonvariational methods during the execution
	when they are shared
	by all configurations (with variability context \emph{true}) in which
	all parameters have only a single concrete value.

% \end{itemize}

In our approach, we transform all methods possible, including libraries,
to push the environment barrier as far outside as possible.
In the JVM, the environment barrier often manifests as native methods, i.e., methods
that are hard-coded in the JVM in other programming languages such as C and C++.
We maintain a list of model classes and side-effect free methods that are
automatically applied when encountered. For all remaining calls to
nonvariational code, we issue warnings during transformation and abort
the execution at runtime when invoked.
We then manually and incrementally inspect aborts in our executions and mark methods as side-effect free
or develop model classes as needed.
In fact, so far, we needed to implement model classes only for a small number of classes.
We have not yet encountered executions that heavily rely on variational interactions
with the environment and thus require additional model classes.
% To mitigate the effort of
% implementing model classes, we adopt an incremental approach. If our tool
% detects that a model class is required for invoking a native method, we insert
% \emph{runtime} checks to monitor if all parameters are conditional value that
% has only one possible concrete value. In this case, the native method is invoked
% without variation, so it is safe to proceed. Otherwise, we reject variational
% execution and investigate how to model the environment with variations. These
% simple runtime checks turn out to be sufficient in all of our subject systems,
% as it is not common that I/O behaviors differ in different variability contexts.

Return instructions are more straightforward to transform than method invocation
instructions. To not prematurely end the execution of a method at a return
instruction, we rewrite the method to use a single return instruction at the end
of the method. If the method being transformed has more than one return
instructions, we rewrite all of them to jump to a single return at the end of the
transformed method.
If necessary, we store the values of different original return instructions in a variable.
Technically, we again replace all non-void returns by a single \instr{areturn}
instruction, returning a reference to the resulting conditional value.
% This way, all original exits of the method are reduced
% to control flow decisions that jump to the only exit.
For example, see
how Line 21 and~23 are transformed to Lines 62, 67 and~70 in Figure~\ref{fig:running}.

\subsection{Using Objects}
\label{subsec:object}

In the JVM, both \emph{class instances} and \emph{arrays} are objects, but the
JVM creates and manipulates class instances and arrays using distinct sets of
instructions. This section presents our transformation of them respectively.
\looseness=-1

\paragraph{Class Instances}
We transform all fields of a class instance to have the conditional value type.
The key idea is to maximize sharing of data across similar class instances. If
two instances of the same class only differ in one field, we represent the difference in
a conditional value for that field, rather than as a conditional
reference to two copies of the object. This design stores variability as local
as possible to avoid redundancy in memory and in computations~\cite{MWK+:ASE16}.
As fields store conditional values (\ref{i:local}), reads and writes to fields
work just as loads and stores to local variables.
\looseness=-1

A technical challenge to independent transformation of bytecode instructions
arises for the \instr{new} instruction used to instantiate classes and push
them to the operand stack. The challenge is that the \instr{new} instruction
creates an uninitialized object that cannot be passed as a reference for safety reasons until
the object's constructor is invoked on it, and thus cannot be wrapped in
a \texttt{V} type as needed for \ref{i:stack}.
Instead, we treat \instr{new} and the subsequent initialization
sequence as \emph{one bytecode instruction} for our transformation.
Whenever we encounter a \instr{new} instruction, we use a
data-flow analysis to identify the relevant following initialization
sequence, re-arranging the original bytecode if necessary to separate
object initialization from other instructions (e.g., instructions to compute constructor parameters).

\paragraph{Arrays} For a given array, we transform it into an array of
conditional values to again store variability as local as possible to preserve sharing.
To support arrays of different length though and fulfill our invariants,
we support also variations of arrays.
That is, an array of objects (\texttt{Object[]}) would be represented
as a conditional array of conditional objects (\texttt{V<V<Object>[]>}).
Type erasure in Java complicates the implementation, but this can be solved by
inserting additional dynamic type checks.
\looseness=-1

We arrived at this design after considering several tradeoffs:
Our representation can store variability more locally, avoiding
that a single variation in an entry requires to copy the entire array;
also load and store operations are simple and fast.
Overheads are only encountered for arrays with different length
in different configurations, which is less common than variability
in values in our experience.
An alternative design could loosen our invariants for arrays and
create a single maximum-length array of conditional values
(based on the length of the configuration with the longest array; \texttt{V<Object>[]})
and a shadow variable and extra instructions for bookkeeping and length checking,
but we only expect marginal performance benefits from this more complicated design.
\looseness=-1

% Although the array of \texttt{V}s representation is sufficient to represent
% variations of equally sized arrays, variably-sized variations require further
% special handling. We represent arrays with variations of different sizes as a
% choice between variational arrays: \texttt{V<V[]>>}. This way, calculating the
% array length under a given variability context is straightforward.  An
% alternative way is to create an array of conditional values with the maximum
% length of all possible lengths, and couple the array with an extra conditional
% value to bookkeep the actual lengths under different contexts. However, this
% design makes iterating arrays more complicated, which we expect to happen very
% often.

\section{Control Transfer}
\label{sec:control}

After describing how to transform bytecode instructions within a given
variability context, we now focus on how to transform control-flow related
constructions that may change variability contexts by splitting or joining executions.
For example, in a branching statement the condition may differ among configurations,
such that we may need to execute both
branches under corresponding variability contexts, but join afterward
to maximally share subsequent executions.
% This phenomenon calls for a more
% advanced control transfer mechanism that supports (1)~jumping under some
% contexts and (2)~not jumping under the other contexts.  In other words, we
% should be able to correctly \emph{split} the execution.  After splitting
% execution, we want to \emph{join} the execution again so that common parts among
% split executions can be shared maximally.

We significantly change the way programs are executed and track and change variability
contexts. As introduced in Section~\ref{sec:motivation}, variability contexts
are propositional formulas over configuration options that describe the partial configuration
space for which an instruction is executed, similar to path conditions in symbolic execution.
Instructions executed in a variability context only have an effect on the state of
that partial configuration space, as discussed, for example, for store instructions in
Section~\ref{subsec:basic}. The challenge is now to propagate and change
variability context to achieve a shared execution for all configurations with maximal sharing.

In this section, we explain how we structure the program in blocks with the same
variability context, and how we transfer control and contexts among these blocks.
Subsequently, we then discuss two important properties of our design:
(1) that variational execution preserves behavior of the original program
(\emph{Correct Execution Property}) and (2) that control transfer among blocks is
efficient (\emph{Optimal Sharing Property}).  Finally, we present some technical
challenges and their solutions regarding stack values during control transfer.

\subsection{VBlock}

We group all instructions that are statically guaranteed to always share
the same variability context at runtime in a \emph{VBlock}.
VBlocks are separated by \emph{conditional} jumps, that is, jumps that
may depend on conditional values, in which case we may ``split'' the execution.
After executing multiple VBlocks we may ``join'' the execution in another
VBlock with a broader variability context (such joining is rare in symbolic
execution approaches).
For example, String replacement of a smiley image (Line 9) in
Listing 1 has a more restricted context than the \texttt{getHTMLHeader} call
(Line 6) because Line 9 is only executed when \texttt{SMILEY} is \texttt{true},
whereas the later \texttt{getHTMLFooter} call is again shared among all configurations.
\looseness=-1

VBlocks are similar to basic blocks in traditional program analyses. However, unlike basic
blocks, which group individual instructions together because they are always
executed in sequence, VBlocks group basic blocks together because they always
share the same variability context. Thus, there can be jumps inside a VBlock
as long as they do not depend on conditional values and thus share the same variability context.

Bytecode instructions can be partitioned into VBlocks by merging basic blocks in
a control-flow graph iteratively until a fixpoint is reached.
A block $B_1$ can be merged with a successor $B_2$ if the jump between
$B_1$ and $B_2$ is not conditional (e.g., \instr{goto} or \emph{if} statement
with non-conditional expression)\footnote{
	\ref{i:stack} implies that all values evaluated
	in an \emph{if} statement are conditional, however, as we will discuss
	later in Section~\ref{sec:optim}, we can optimize the transformation
	to statically recognize values that will not depend on configuration
	options, including, in the simplest case, constants.
} and all predecessors of $B_2$ are in the same VBlock. The latter condition
is needed to recognize potential join points, when a block can be reached from
two different VBlocks.
Hence, a VBlock can be terminated by either
a conditional jump or an unconditional jump. A VBlock can end with 
an unconditional jump if, for example, while
merging basic blocks to form VBlocks, basic block
\texttt{A} has an unconditional jump to basic block \texttt{C}, while basic
block \texttt{B} has a conditional jump to \texttt{C}. We cannot merge
\texttt{A} and \texttt{C} into one VBlock because of the conditional jump
from \texttt{B}. Thus, \texttt{A} and \texttt{C} have to be separated into
two different VBlocks with an unconditional jump between them.

\subsection{Execution Strategy}\label{subsec:execution}

This subsection presents how VBlocks are used. We first outline the goals of
using VBlocks to achieve splitting and joining execution. Then, we present a
solution that achieves our goals and provide an example.
\looseness=-1

\paragraph{Goals}
Whereas variational-execution approaches that modify interpreters
(such as VarexJ) can track multiple instruction pointers and their
variability contexts, we need to
cope with the fact that the instruction pointer of an unmodified JVM can only point
to a single location at a time. So instead of changing the control transfer
mechanism of the JVM, we use VBlocks to organize and create the execution order
we want. At a high level, we pursue the following:

\begin{itemize}

	\item Both branches of a conditional jump can be executed under
		corresponding restricted contexts (we call them ``subcontexts''). That is,
		we are able to split execution.

	\item The code after both branches of a conditional jump should be executed
		only once for mutually exclusive contexts. That is, we should join execution as early as possible.

\end{itemize}

\paragraph{Context propagation} Using VBlocks, we modify control flow decisions and manipulate variability contexts
to achieve splitting and joining. The key idea of our design is to associate each
VBlock with a variability context (a fresh local variable). We dynamically update variability contexts
along execution to keep track of which VBlock(s) can be executed next and under which context.
At any point in a method's execution, all VBlocks with a \emph{satisfiable} variability context
(i.e., the proposition formula is satisfiable) can be executed.
The order in which multiple VBlocks with satisfiable contexts are executed does
not matter for correctness, but does matter for performance, as we will show in Section~\ref{subsec:properties}.
\looseness=-1

At a jump between VBlocks, we transfer the current block's variability context
to the target block's context. If the jump is conditional, we split the
current variability context and transfer the two mutually exclusive contexts to the
two successor VBlocks of the jump. The split is determined by the partial configuration
space in which the \emph{if} statement's expression evaluates to true.

% Since this is the only way to produce new execution contexts, the set of
% execution contexts at any given time is always mutually exclusive. This property
% allows us to execute VBlocks with satisfiable contexts in any order without
% affecting the correctness of the program, since every context logically
% represents a different execution of the program. In
% Section~\ref{subsec:properties} we will prove that this property holds.
% Nonetheless, the order of execution of VBlocks does affect the ``joining'' of
% split contexts, so we must use an ordering strategy for the VBlocks to ensure
% that split executions join as early as possible. In the following, we first show
% our mechanism of using and updating variability contexts. Then, we present an
% ordering scheme for VBlocks when multiple of them are executable.

% To support this execution strategy, we represent methods as sequences of VBlocks
% which propogate their variability context to other VBlocks when executed.

To describe the control transfer more precisely, let us
denote the sequence of VBlocks as $b_0, b_1, \ldots, b_n (n \geq 0)$, where
$b_0$ represents the entry node in the control flow graph and $b_n$ represents
the exit node. Let us denote the variability context of a VBlock  $b_i$ as $\phi(b_i)$  (stored in a fresh local variable for each VBlock).
\begin{itemize}
	\item At the beginning of a method
execution, we initialize $\phi(b_0)$ with the method context, and $\phi(b_i) =
\textit{False}$ for all other VBlocks to indicate that only the initial VBlock of the
method can be executed.

	\item After executing a VBlock $b_i$, we remember its variability context $\Phi=\phi(b_i)$
	and then set that variability to $\textit{False}$,
indicating that this block should not be immediately executed again.
	We subsequently propagate its prior variability context $\Phi$ as follows:

	\begin{enumerate}

\item If the execution of VBlock $b_i$ ends with an unconditional jump (e.g., \instr{goto} instruction) to another VBlock $b_j$, the context of
$b_j$ is updated as a disjunction between the current context of $b_j$ and $b_i$'s prior context $\Phi$. A disjunction is required because the target block may already
have been executable under a different context, which we now broaden to join executions.
\begin{equation}\label{eq:nojump}
	\phi'(b_j) = \phi(b_j) \lor \Phi
\end{equation}

\item If the execution of VBlock $b_i$ ends with a conditional jump with two possible target VBlocks
	$b_j$ and $b_k$,\footnote{We transform switch statements into an equivalent
	series of if-else statements to simplify our design of control transfer.}
	we split the execution based on the condition of the jump
	(usually the top value on the stack representing result of evaluating an \emph{if}
	statement's expression).
	Let us denote the variability context in which the jump condition indicates a jump
	to $b_j$ as $X$. For example, the condition of the first \emph{if} statement
	in our WordPress example is $\langle SMILEY, 1, 0 \rangle$, which indicates the
	\emph{then} branch should be taken under context $X=SMILEY$.
	 We update the variability contexts of $b_j$ and $b_k$ as follows,
	again considering potential joins:
\begin{align}\label{eq:jump}
	\phi'(b_j) &= \phi(b_j) \lor (X \land \Phi) & \phi'(b_k) &= \phi(b_k) \lor
	(\neg X \land \Phi)
\end{align}

\end{enumerate}

	\item After propagating the variability context, the control transfer (i.e., the actual
	instruction pointer in the JVM) does not actually follow
	the jump.

\end{itemize}

\paragraph{Execution Order} The actual execution order though (in terms of moving the instruction pointer) is independent from the transfer of variability
contexts. We start execution at the beginning of the method with $b_0$.
At the end of a VBlock $b_i$, we jump to the next VBlock $b_{i+1}$ by default, even if the block ended with a different jump.
If that VBlock's variability context is unsatisfiable, we proceed to the next VBlock, and so forth.
We only jump back to a VBlock with a lower index (using a plain \instr{goto} instruction) when we update the variability context of an earlier block
to be satisfiable as part of the described context transfer.
This way, the instruction pointer is always at an unsatisfiable block (to be skipped) or at the satisfiable block with the lowest index.
This strategy ensures that later VBlocks are always executed with
joined variability contexts from earlier VBlocks and that VBlock $b_n$ is executed last with
the full method context. For that reason,
the indexing order of VBlocks matters. Figure~\ref{fig:jumpexample} illustrates the idea of jumping among VBlocks with a concrete example.

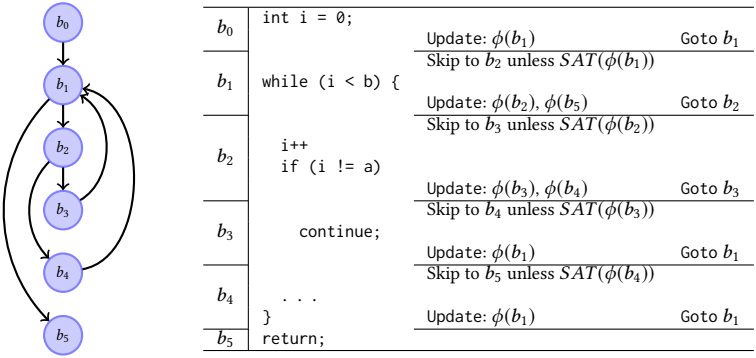
\begin{figure}[!tpb]
	\begin{subfigure}[bt]{0.2\textwidth}
		% \centering
		\begin{tikzpicture}[scale=0.6, transform shape,
			n/.style={circle,draw=blue!50,fill=blue!20,thick,inner sep=1.5mm,minimum size=2mm} ]
			\node[n] (A) {$b_0$};
			\node[n] (B) [below=0.5cmof A] {$b_1$};
			\node[n] (C) [below=0.5cmof B] {$b_2$};
			\node[n] (D) [below=0.5cmof C] {$b_3$};
			\node[n] (E) [below=0.5cmof D] {$b_4$};
			\node[n] (F) [below=0.5cmof E] {$b_5$};
			\draw[->, thick] (A) to (B);
			\draw[->, thick] (B) to (C);
			\draw[->, thick] (C) to (D);
			\draw[->, thick] (C) to[bend right=45] (E);
			\draw[->, thick] (D) to[bend right=60] (B);
			\draw[->, thick] (E) to[bend right=80] (B);
			\draw[->, thick] (B) to[bend right=45] (F);
		\end{tikzpicture}
	\end{subfigure}
	\begin{subfigure}[bt]{0.6\textwidth}
		% \centering
		\scriptsize
		\renewcommand{\arraystretch}{1}
		\begin{tabular}{l | l ll}
				% &\cellcolor{red!25}\\[-0.5em]
			\hline

			\multirow{2}{*}{$b_0$} & \texttt{int i = 0;} & & \\
								   & & \verb|Update|: $\phi(b_1)$ & \verb|Goto| $b_1$\\

			\cline{1-1}\cline{3-4}
				% &\cellcolor{green!25}\\[-0.5em]
			\multirow{3}{*}{$b_1$} & & Skip to $b_2$ unless $SAT(\phi(b_1))$&\\
								   & \texttt{while (i < b) \{} & &  \\
								   & & \verb|Update|: $\phi(b_2)$, $\phi(b_5)$ & \verb|Goto| $b_2$\\
			\cline{1-1}\cline{3-4}
			\multirow{4}{*}{$b_2$} & & Skip to $b_3$ unless $SAT(\phi(b_2))$&\\
								   & \texttt{\quad i++} & & \\
			                       & \texttt{\quad if (i != a)} & & \\
								   & & \verb|Update|: $\phi(b_3)$, $\phi(b_4)$ & \verb|Goto| $b_3$\\
			\cline{1-1}\cline{3-4}
			\multirow{3}{*}{$b_3$} & & Skip to $b_4$ unless $SAT(\phi(b_3))$&\\
								   & \texttt{\quad\quad continue;} & & \\
								   & & \verb|Update|: $\phi(b_1)$ & \verb|Goto| $b_1$ \\%$\begin{cases} b_1 & \text{if } SAT(\phi(b_1)) \\ b_4 & \text{otherwise} \end{cases}$\\
			\cline{1-1}\cline{3-4}

			\multirow{3}{*}{$b_4$} & & Skip to $b_5$ unless $SAT(\phi(b_4))$&\\
								   & \texttt{\quad \ldots} &  & \\
								   & \texttt{\}}  & \verb|Update|: $\phi(b_1)$ & \verb|Goto| $b_1$ \\ % $\begin{cases} b_1 & \text{if } SAT(\phi(b_1)) \\ b_5 & \text{otherwise} \end{cases}$\\
			\cline{1-1}\cline{3-4}
			\multirow{1}{*}{$b_5$} & \texttt{return;} & &  \\
			\hline
		\end{tabular}
	\end{subfigure}
	\caption{\footnotesize An example illustrating control-flow encoding through updates of variability contexts and jumps between blocks.}
	\label{fig:jumpexample}
\end{figure}

%  If propagating into one VBlock, the
% context is disjoined (i.e., join with OR) with the next VBlock in the CFG,
% indicating that the next VBlock is ready to be executed. If propagating into two
% VBlocks, which means a conditional jump is encountered, the context is split
% into two \emph{disjoint} narrowed contexts based on the jump condition. We
% formalize our mechanism of updating contexts as follow:

\paragraph{Ordering VBlock Execution}
Given that we always execute the first VBlock with a satisfiable variability context
and always join at later VBlocks, we can execute the same method in different
ways by changing the way we order the VBlocks.
We can reorder VBlocks in different orders as long as the first and last VBlock remain constant
(the last block ending with a return statement must be executed last) and always
achieve equivalent (i.e., correct) results, as we will show in
Section~\ref{subsec:properties}.
However, as the block order determines the join points, different orders
may be more or less effective at joining early and sharing subsequent computations.

To maximize sharing during the execution (i.e., prefer executing a block
once under a broader variability context rather than multiple times under
narrow contexts), we order VBlocks based on the \emph{strict transitive predecessor}
relation in the control-flow graph. A VBlock $b_i$ is a strict transitive
predecessor of $b_j$ if there is a path from $b_i$ to $b_j$ in control-flow
graph, but not from $b_j$ to $b_i$ (i.e., not in a loop).
\textit{For any pair of VBlocks, if one VBlock is a strict transitive
	predecessor of the other,
	the transitive predecessor shall have the lower VBlock index to be executed first.}
For other pairs, we preserve the original lexical order produced by the compiler as
a default.

In the next subsection, we will show that the above partial order is
sufficient to statically guarantee optimal sharing on a subset of control-flow graphs,
regardless of the original lexical order of the bytecode, but that optimality
cannot be statically guaranteed for all control-flow graphs. We will also
experimentally show in Section~\ref{sec:eval} that this order is nearly always optimal for the remaining
control-flow graphs.

\paragraph{Example} Let us exemplify our solution by stepping through the
\texttt{getWeather} method in Listing 2. There are four VBlocks: code before the
\texttt{if} statement ($b_0$, Line~57-59), then branch ($b_1$, Line~60-64), else
branch ($b_2$, Line~65-69) and
return block ($b_3$, Line~70). 
These blocks are already indexed according to the strict transitive predecessor relation:
$b_0$ is executed first, $b_1$ and $b_2$ are executed
before $b_3$; the order between $b_1$ and $b_2$ is merely
derived from the lexical order and could be switched.
 Initially, $\phi(b_0) = \textit{MCtx}\ \text{(method context)}$
and $\phi(b_1)=\phi(b_2)=\phi(b_3) = \textit{False}$.  After executing $b_0$ at Line~59, $\phi(b_1)$
and $\phi(b_2)$ are updated to $\phi(b_1) = \textit{False} \lor (\textit{FAHRENHEIT} \land \textit{MCtx})$
and $\phi(b_2) = \textit{False} \lor (\neg \textit{FAHRENHEIT} \land \textit{MCtx})$, thus
\emph{splitting} the execution. Note that this update of contexts is not shown
in Listing~2 because we transform the control flow in bytecode differently from
how we show for Java. 
At this point, both $\phi(b_1)$ and
$\phi(b_2)$ are satisfiable and execution continues with the next VBlock $b_1$.
After executing $b_1$ at Line~64, $\phi(b_3)$ is updated to $\phi(b_3) =
\textit{False} \lor (\textit{FAHRENHEIT} \land \textit{MCtx})$ because $b_3$ is the sole successor of $b_1$ in
the control flow graph. We execute the next satisfiable block, which is $b_2$,
after which $\phi(b_3)$ is updated to
$\phi(b_3) = (\neg \textit{FAHRENHEIT} \land \textit{MCtx}) \lor (\textit{FAHRENHEIT} \land \textit{MCtx}) = \textit{MCtx}$;
thus, $b_3$ at Line~70 is executed last under the \emph{joined} context $\textit{MCtx}$.

\subsection{Properties}\label{subsec:properties}

We have presented how we choose VBlocks for execution. While splitting executions,
we need to ensure that the execution order is correct. By always executing
the satisfiable VBlock with the lowest index first and ordering VBlocks deliberately,
we make sure that the joining happens as early as possible. This section formalizes these properties.

\paragraph{Correctness} The following property ensures that our variational
execution is correct, in a sense that it preserves the semantics of the original
program.

\begin{property*}[Correct Execution Property]\label{property:correct}
	At any point of execution, if there are multiple VBlocks with satisfiable
	contexts, the order in which they are executed does not affect correctness of execution.
	\looseness=-1
\end{property*}

To prove this, we first introduce a useful lemma:

\begin{lemma*}[Disjoint Context Lemma]
	At any point of variational execution, the context of two different VBlocks
	are mutually exclusive. That is, $\phi(b_i) \land \phi(b_j) = \textit{False}$ for any
	$i \neq j$.
\end{lemma*}

Mutual exclusion is guaranteed by the way we propagate contexts in Equations \ref{eq:nojump} and~\ref{eq:jump}.
A proof by induction can be found in the appendix. With this lemma, we prove our \emph{Correct Execution Property} as
follows:

\begin{proof}
	Ensuring that VBlocks have mutually exclusive variability contexts guarantees that each VBlock
	operates on mutually exclusive runtime states. As we have discussed in
	Section~\ref{sec:trans}, states (e.g., local variables, fields) are stored
	separately for different contexts using conditional values. Our variability
	contexts further ensure that all state changes only update values in the
	(disjoint) contexts referred to by the current variability context. Thus,
	execution order among satisfiable VBlocks does not affect correctness of
	overall variational execution.
\end{proof}

\paragraph{Optimal sharing} The main utility of variational execution is its
ability to share common computations; our execution scheme pursues to perform
executions with the broadest variability context possible.
While we cannot share repeated executions under the same context, we can
avoid executing the same VBlock under mutually exclusive contexts and
rather execute it once, shared, under a broader context.
In a nutshell, what we want to achieve is to
execute every VBlock as few times as possible by sharing the execution of
VBlocks in different contexts. This sharing is crucial for the overall
performance of variational execution, otherwise it may degrade to
executing each variation in a brute-force way or sharing only common prefixes of traces, conceptually equivalent to
joining only after the very last instruction.

In order to formalize \emph{optimal sharing}, we define a variational
trace as a chronological sequence of VBlocks executed during variational
execution.
We denote a variational trace as a sequence of executed VBlocks with corresponding
variability context, e.g., $t_v = [b_0^{True}, b_1^{\alpha},
b_2^{\neg \alpha}, b_3^{True}]$. Conceptually, a variational
trace corresponds to a separate concrete trace for each configuration,
in our example $t_{\alpha}$ =
[$b_0, b_1, b_3]$ and $t_{\neg \alpha}$=[$b_0, b_2, b_3]$.
Another variational
execution trace that represents the same concrete traces could be
$t_v'=[b_0^{T}, b_1^{\alpha}, b_3^{\alpha}, b_2^{\neg \alpha}, b_3^{\neg
\alpha}]$. It is likely that $t_v$ is more efficient than $t_v'$ because $b_3$
is executed twice in $t_v'$.

A variational trace can be seen as the result of aligning multiple concrete
traces. Different aligning schemes produce different variational traces (e.g.,
$t_v$ and $t_v'$). Given a set of concrete traces, we can use sequence alignment
algorithms (e.g., Needleman-Wunsch algorithm~\cite{Needleman:1970gm}) to obtain
a globally optimal solution of merging concrete traces. For example, two optimal
matchings of $t_{\alpha}$ and $t_{\neg \alpha}$ are $t_o=[b_0, b_1, b_2, b_3]$ and
$[b_0, b_2, b_1, b_3]$.  We use $length(t)$ to denote the number of elements in
a trace. For example, $length(t_{v})=4$, and $length(t_v')=5$.

\begin{definition}[Optimal Sharing]\label{def:optimal}
	% Every VBlock is executed under the largest possible context.
	Given a variational trace $t_v$ and its corresponding set of concrete traces
	$t_1, t_2, \ldots, t_m$, we say $t_v$ has optimal sharing if and only if
	$length(t_v) = length(t_{o})$, where $t_{o}$ is the
	optimal matching of $t_1, t_2, \ldots, t_m$.
\end{definition}

% It would be ideal if optimal sharing could be achieved for all possible programs
% in the wild, but there is no join strategy that could statically order blocks to guarantee
% optimal sharing for all executions of all programs.
% Figure~\ref{fig:counter_example} illustrates an example: In order to
% achieve optimal sharing (with optimal defined as the optimal trace alignment in the figure),
% there is both a case where $D$ needs to be executed before $C$ and a case
% where $C$ needs to be executed before $D$ after a control-flow decision
% at $B$ (critical nodes highlighted in the trace).
% That is, we cannot statically decide an ordering between $C$ and $D$,
% and even an optimal decision at runtime would have to depend on knowing
% the \emph{future} execution trace.
% We could apply
% some greedy strategies to approximate optimality, but that the required runtime
% monitoring is unlikely to justify the performance benefits of additional sharing
% execution.

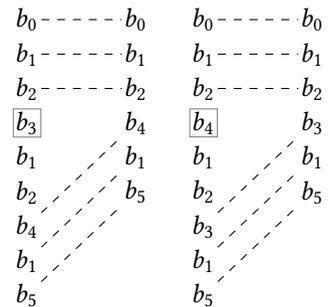
\begin{wrapfigure}{r}{0.4\textwidth}
% \begin{figure}[!tpb]
% \begin{centering}
	% \begin{tikzpicture}[scale=0.6, transform shape,
	% 	n/.style={circle,draw=blue!50,fill=blue!20,thick,inner sep=2mm,minimum size=6mm} ]
	% 	\node[n] (A) {A};
	% 	\node[n] (B) [below=of A] {B};
	% 	\node[n] (C) [below=of B, xshift=-1.5cm] {C};
	% 	\node[n] (D) [below=of B, xshift=1.5cm] {D};
	% 	\node[n] (E) [below=of C, xshift=1.5cm] {E};
	% 	\draw[->, thick] (A) to (B);
	% 	\draw[->, thick] (B) to[bend left=45] (C);
	% 	\draw[->, thick] (C) to[bend left=45] (B);
	% 	\draw[->, thick] (B) to[bend left=45] (D);
	% 	\draw[->, thick] (D) to[bend left=45] (C);
	% 	\draw[->, thick] (D) to (E);
	% \end{tikzpicture}
	% \begin{tikzpicture}[scale=0.6, transform shape,
	% 	n/.style={circle,draw=blue!50,fill=blue!20,thick,inner sep=2mm,minimum size=6mm} ]
	% 	\node[n] (A) {$b_0$};
	% 	\node[n] (B) [below=0.5cmof A] {$b_1$};
	% 	\node[n] (C) [below=0.5cmof B] {$b_2$};
	% 	\node[n] (D) [below=0.5cmof C] {$b_3$};
	% 	\node[n] (E) [below=0.5cmof D] {$b_4$};
	% 	\node[n] (F) [below=0.5cmof E] {$b_5$};
	% 	\draw[->, thick] (A) to (B);
	% 	\draw[->, thick] (B) to (C);
	% 	\draw[->, thick] (C) to (D);
	% 	\draw[->, thick] (C) to[bend right=45] (E);
	% 	\draw[->, thick] (D) to[bend right=60] (B);
	% 	\draw[->, thick] (E) to[bend right=80] (B);
	% 	\draw[->, thick] (B) to[bend right=45] (F);
	% \end{tikzpicture}
	% \hspace{2cm}
	\centering
	\begin{tikzpicture}[scale=1, transform shape, node distance=2pt, inner
		sep=1pt]
		\node (11) {$b_0$};
		\node (12) [below=of 11] {$b_1$};
		\node (13) [below=of 12] {$b_2$};
		\node (14) [rectangle, draw=black!50,below=of 13] {$b_3$};
		\node (15) [below=of 14] {$b_1$};
		\node (16) [below=of 15] {$b_2$};
		\node (17) [below=of 16] {$b_4$};
		\node (18) [below=of 17] {$b_1$};
		\node (19) [below=of 18] {$b_5$};

		\node (21) [right=of 11, xshift=1cm] {$b_0$};
		\node (22) [below=of 21] {$b_1$};
		\node (23) [below=of 22] {$b_2$};
		\node (24) [below=of 23] {$b_4$};
		\node (25) [below=of 24] {$b_1$};
		\node (26) [below=of 25] {$b_5$};
		% \node (27) [below=of 26] {C};
		% \node (28) [below=of 27] {B};
		% \node (29) [below=of 28] {D};
		% \node (210) [below=of 29] {E};

		\draw[-,dashed] (11) to (21);
		\draw[-,dashed] (12) to (22);
		\draw[-,dashed] (13) to (23);
		% \draw[-,dashed] (14) to (25);
		% \draw[-,dashed] (15) to (26);
		% \draw[-,dashed] (16) to (27);
		\draw[-,dashed] (17) to (24);
		\draw[-,dashed] (18) to (25);
		\draw[-,dashed] (19) to (26);
	\end{tikzpicture}
	\hspace{1em}
	\begin{tikzpicture}[scale=1, transform shape, node distance=2pt, inner
		sep=1pt]
		\node (11) {$b_0$};
		\node (12) [below=of 11] {$b_1$};
		\node (13) [below=of 12] {$b_2$};
		\node (14) [rectangle, draw=black!50,below=of 13] {$b_4$};
		\node (15) [below=of 14] {$b_1$};
		\node (16) [below=of 15] {$b_2$};
		\node (17) [below=of 16] {$b_3$};
		\node (18) [below=of 17] {$b_1$};
		\node (19) [below=of 18] {$b_5$};

		\node (21) [right=of 11, xshift=1cm] {$b_0$};
		\node (22) [below=of 21] {$b_1$};
		\node (23) [below=of 22] {$b_2$};
		\node (24) [below=of 23] {$b_3$};
		\node (25) [below=of 24] {$b_1$};
		\node (26) [below=of 25] {$b_5$};
		% \node (27) [below=of 26] {C};
		% \node (28) [below=of 27] {B};
		% \node (29) [below=of 28] {D};
		% \node (210) [below=of 29] {E};

		\draw[-,dashed] (11) to (21);
		\draw[-,dashed] (12) to (22);
		\draw[-,dashed] (13) to (23);
		% \draw[-,dashed] (14) to (25);
		% \draw[-,dashed] (15) to (26);
		% \draw[-,dashed] (16) to (27);
		\draw[-,dashed] (17) to (24);
		\draw[-,dashed] (18) to (25);
		\draw[-,dashed] (19) to (26);
	\end{tikzpicture}
% \end{centering}
\caption{\footnotesize An example where static order between VBlocks cannot not achieve optimal
sharing. The control-flow graph is shown in Figure~\ref{fig:jumpexample}}

\label{fig:counter_example}
% \end{figure}
\vspace*{-2em}
\end{wrapfigure}

It would be ideal if optimal sharing could be achieved for all possible
programs in the wild, but there is no join strategy that could \emph{statically} order
blocks to guarantee optimal sharing for all executions of all programs.
Figure~\ref{fig:counter_example} illustrates an example: In order to achieve
optimal sharing (with optimal defined as the optimal trace alignment in the
figure), there is both a case where $b_3$ needs to be executed before $b_4$ and a
case where $b_4$ needs to be executed before $b_3$ after a control-flow decision at
$b_2$ (critical nodes highlighted in the trace).  That is, we cannot \emph{statically}
decide an ordering between $b_3$ and $b_4$, and even an optimal decision at runtime
would have to depend on knowing the \emph{future} execution trace.  We could
apply some greedy strategies to approximate optimality, but that the required
runtime monitoring is unlikely to justify the performance benefits of
additional sharing execution.

% There are two problems with deciding execution order of VBlocks dynamically. On
% the one hand, optimal order is not decidable without knowing the future
% execution, as we have shown above with a counter example. We could apply
% some greedy strategies to approximate optimality, but that requires runtime
% monitoring of execution to make greedy decisions. The overhead of monitoring
% values at runtime is likely to offset the performance benefits of sharing
% execution.

Fortunately, we can prove optimal sharing for static VBlock ordering for many
shapes of control-flow graphs and will show in our empirical
evaluation that the remaining ones (often with nontrivial interleaving of
looping and branching instructions)  are often optimal for actual executions.

\begin{property*}[Optimal Sharing Property]\label{property:optimal}
	Given a control flow graph where each node represents a VBlock, our
	variational execution based on the strict transitive predecessor relation on this graph has optimal sharing if it is acyclic or
	only contains simple loops. A loop is a simple loop if it satisfies the
	following three criteria: (1) has only one loop header; (2) has only one
	exiting node; (3) has no conditional jumps among nodes in the loop.
\end{property*}

The proof can be found in the appendix. Intuitively, we prove by case
analysis that our variational trace has the same length as the optimal alignment
of corresponding concrete traces in every possible case. Since we only consider
simple control-flow graphs, the length of our variational trace and the length
of the optimal alignment trace can be determined from the structure of the control-flow graph.

% The \emph{Optimal Sharing Property} ensures that our partial execution order of
% VBlocks is efficient.
% Note that, for a given program, it is possible to have
% multiple variational traces that satisfy our definition of optimal sharing. We
% assume that they have similar performance by ignoring the size of each
% individual VBlock. Although this performance guarantee is not made on all
% possible CFGs, we found it to be sufficient for most programs in practice, as we
% will show in Section~\ref{sec:eval}.

\subsection{Values on the Stack between VBlocks}

\begin{wrapfigure}{r}{0.33\textwidth}
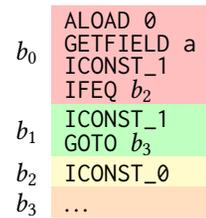


	% \begin{subfigure}[b]{0.4\textwidth}
	% 	\begin{lstlisting}[tabsize=2,mathescape=true]
% public boolean foo{
	% return this.a == 1;
% }
	% 	\end{lstlisting}
	% 	\caption{Caption code}
	% \end{subfigure}
	% \begin{subfigure}[b]{0.5\textwidth}
		\centering
		\renewcommand{\arraystretch}{0.6}
		\begin{tabular}{ l  l }
				&\cellcolor{red!25}\\[-0.5em]

				\multirow{4}{*}{$b_0$} & \cellcolor{red!25}\texttt{ALOAD 0}\\
				& \cellcolor{red!25}\texttt{GETFIELD a}\\
				& \cellcolor{red!25}\texttt{ICONST\_1}\\
				& \cellcolor{red!25}\texttt{IFEQ $b_2$}\\

				&\cellcolor{green!25}\\[-0.5em]

				\multirow{2}{*}{$b_1$} & \cellcolor{green!25}\texttt{ICONST\_1}\\
				& \cellcolor{green!25}\texttt{GOTO $b_3$}\\

				&\cellcolor{yellow!25}\\[-0.5em]

				$b_2$ & \cellcolor{yellow!25}\texttt{ICONST\_0}\\

				&\cellcolor{orange!25}\\[-0.5em]

				$b_3$ & \cellcolor{orange!25}\ldots\\

		\end{tabular}
	% \caption{Caption Table}
	% \end{subfigure}
	\caption{\footnotesize A snippet of bytecode showing the scenario where VBlocks could
	leave some values on the operand stack after execution.}
	\label{fig:UnbalanceStack}
	\vspace*{-2em}
\end{wrapfigure}

In Java, blocks can leave values on the operand stack to be
consumed by subsequent blocks. Since, in variational execution,
there might be multiple successor blocks that will be executed,
and successor blocks may not be executed immediately after their
predecessor, sharing values on the stack becomes tricky.
Since the operand stack in a commodity JVM is not variational itself,
we cannot pop the same value from the stack under different variability contexts
as possible when modifying the interpreter itself (e.g., done in
VarexJ~\cite{MWK+:ASE16}).

Figure~\ref{fig:UnbalanceStack} shows a concrete example in which VBlock $b_0$
leaves some value on the operand stack after execution, that both blocks $b_1$
and $b_2$ try to read. Conversely, those two blocks each leaves a (different) value on the stack
that $b_3$ attempts to consume.
% Let us suppose the field
% \texttt{a} represents some variation in the program.  According to our previous
% definition of VBlock, this snippet of bytecode is divided into 4 VBlocks,
% highlighted with different colors.  In the first VBlock, it compares the value
% of the field \texttt{a} with constant $1$.  If the comparison succeeds, the
% execution jumps to $b_2$. Otherwise, execution continues at $b_1$. Both $b_1$
% and $b_2$ flow into $b_3$, where we return the result.

%\begin{figure}[!tpb]
%	\centering
%	\includegraphics[width=\linewidth]{img/UnbalanceStack}
%	\caption{None-empty stack example. The left side shows Java code and the
%	right side shows the corresponding bytecode compiled using Oracle JDK 1.8}
%	\label{fig:UnbalanceStack}
%\end{figure}

% As we can see, both $b_1$ and $b_2$ leave a value on the operand stack, and that
% value will be consumed by $b_3$ eventually. Suppose after the execution of
% $b_1$, we decide to execute another VBlock, say $b_4$. Let us further assume
% that $b_4$ expects two values on the operand stack. In this scenario, the
% transition from $b_1$ to $b_4$ becomes broken because $b_4$ expects more values
% on the operand stack than the one value pushed by $b_1$, not to mention the fact
% that $b_4$ is expecting different values.

To make the transition between VBlocks possible in all cases, we need one more
invariant:

\begin{enumerate}[label=\textit{Invariant \arabic*}, align=left, resume]
	\item\label{i:boundary} A VBlock does not leave any values on the operand stack at
		the VBlock boundary.
\end{enumerate}

To meet this invariant, we store all remaining values on the operand stack (if
they exist) to local variables, at the end of each VBlock. Then at the beginning
of each VBlock, we check if the current VBlock expects some values from the
operand stack, and load those values from corresponding local variables if
so. Since we support loading and storing under different variability contexts,
as discussed above, this solution generalizes to all control-flow graphs.

\section{Implementations, Optimizations, Limitations}
\label{sec:optim}

We implemented a bytecode transformation tool for the ideas discussed in
Section~\ref{sec:trans} and Section~\ref{sec:control}, and we call it
\emph{VarexC}.  We use the ASM\footnote{\url{http://asm.ow2.org/}} library to
implement our data flow analysis and transformation of bytecode. Transformations
happen at class loading time via our own class loader that transforms classes
before they are actually loaded. We also save the previously transformed classes
and reuse them if there are no changes. To ensure correct implementation of
variational execution, we apply differential testing to compare execution results and
execution traces in variational execution against brute-force concrete
executions for our subject systems~\cite{K:17}. 
Our implementation is available on
GitHub.\footnote{\url{https://github.com/chupanw/vbc}}
% The choice of bytecode engineering library is
% orthogonal to our bytecode transformation approach. 
We implement transformations for all bytecode instructions and provide
a mechanism for model classes, as discussed in
Section~\ref{subsec:method}. 
In addition to implementing the full transformation described previously,
we explored two optimizations and briefly outline them.
% Since they are orthogonal to the key concepts of our approach, we briefly outline them separately in this section. 
Finally, we discuss the current limitations of our tool.\looseness=-1

% We have presented how bytecode can be transformed to perform variational
% execution. % For example, \emph{not all} instructions need to be transformed, and \emph{some
% classes} can be implemented in a more efficient way to handle variational data.
% This section shows several ways of optimizing the performance of our transformed
% bytecode.

\subsection{Optimization: Deciding What to Transform}\label{subsec:dfa}

Not all bytecode instructions in a program may depend on configuration
options. If we can statically guarantee that parts of the program never
depend on conditional values or conditional jumps, we can reduce our
transformation to relevant parts; this reduces the overhead of computing with conditional values and conditional jumps where not needed.
Guaranteed non-conditional computations happen often in the beginning
of methods and typically involve initialization sequences or constants,
such as in logging statements ``\texttt{System.out.println("done");}''.

We designed a simple data-flow analysis to decide which instructions need to
be transformed. Along the lines of a standard taint analysis, we mark all
local variables and values on the operand stack as conditional or
unconditional with a `lift' bit,  marking them as conditional when
instructions based on other conditional values write to them.

So far, we implemented an intra-procedural analysis that assumes all fields
(including fields representing configuration options, as in our WordPress
example) and method parameters and method results are conditional.
As such,
all stack values produced by field reads, loads of method parameters, and
results from method invocations are marked with the lift bit. We then
propagate the lift bit to all values resulting from computations in
which operands had the lift bit and to local variables when such values
are stored.
Based on the lift bit, we decide which control-flow decisions are
conditional (i.e., potentially depend on conditional values) and
compute VBlocks correspondingly.
Finally, we determine with a simple control-flow analysis,
which VBlocks are guaranteed to be executed with the method's variability
context (in a nutshell, all VBlocks that dominate the method exit) and mark all variables
stored in other VBlocks as conditional, as they may be stored only
in restricted contexts. As is common for data-flow analyses, we repeat
these computations until a fixpoint is reached.

Based on our analysis, we transform bytecode based on a potentially smaller
number of VBlocks (because some jumps are statically guaranteed to be
non-conditional
when their expression does not have the lift bit). We also
transform only variables and instructions with the lift bit. We introduce
additional instructions to  translate concrete values into conditional values
when values flow from unmodified into transformed code (i.e., just wrapping
the concrete value in a \texttt{V} instance, boxing primitive types if
necessary), such that our invariants still hold from the perspective of the
transformed instructions.
\looseness=-1

Our current analysis is very conservative, because it assumes all fields and
method signatures are conditional, thus most savings relate to constants and
initialization sequences. Nonetheless, in the programs of our evaluation
(Sec.~\ref{sec:eval}), we can statically decide to not lift up to $32.6$~percent
of all instructions, which however has only a marginal impact on performance.
We hope that future work can push this analysis even further by performing an
inter-procedural analysis to determine which methods and method arguments need
to be transformed, potentially providing multiple transformed or partially
transformed copies of the same method.

\subsection{Optimization: Using Model Classes}

As discussed in Section~\ref{subsec:method}, we provide a mechanism for  model classes with which we can implement custom implementations for
classes where automated transformations are not possible (e.g., native
methods, environment barrier) or inefficient.
In fact, it is often possible to provide more efficient implementations
of common data-structure implementations that are specifically designed for
variability~\cite{Meng:2017bx,Walkingshaw:2014ch}. Our model-class mechanism
allows drop-in replacements for such classes.

% to provide a
% mechanism that allows certain classes to be \emph{untransformed}. These classes
% are implemented to handle variational data explicitly. Model classes are useful
% for dealing with classes that we cannot transform directly, but more
% importantly, they are useful for optimizing performance. Walkingshaw et al. have
% shown promising results of performance improvement when common data structures
% like stacks are designed to handle variational data
% efficiently~\cite{Meng:2017bx,Walkingshaw:2014ch}. In this section, we show two
% ways of using model classes to improve performance: drop-in replacement and
% aggressive replacement.

\paragraph{Variational data structures}  We implemented a small number of
custom variational data structures for commonly used collections. For example,
instead of an automated transformation of the \texttt{java.util.LinkedList}
class, which would support conditional values and conditional successors of
linked-list nodes, we use a custom implementation that internally stores a
list of optional elements and provides corresponding accessor functions.
Similarly, rather than automated transformation of \texttt{java.util.HashSet}
objects, we can represent variational sets as a mapping from values to variability
contexts that describe the configuration space in which the set contains that
value. As explored by Walkingshaw et al., such tailored representations are
often (though not generally) much more efficient, especially when they hold
many optional elements with different conditions
\cite{Meng:2017bx,Walkingshaw:2014ch}.

% A model class can be implemented as a drop-in
% replacement for its corresponding transformed class. For example, our
% transformation transforms \texttt{java.util.LinkedList} as a list of optional
% elements.  We can instead implement a model class for \texttt{LinkedList} to
% store variation differently. There could be different ways of implementing a
% model class for \texttt{LinkedList}. One naive way is to store a separate
% concrete \texttt{LinkedList} for every mutually exclusive variability context.
% This way, model \texttt{LinkedList} is essentially a wrapper for concrete
% \texttt{LinkedList}s in corresponding contexts. Operations on this model
% \texttt{LinkedList} are mapped to concrete \texttt{LinkedList} using
% \texttt{smap} or \texttt{sflatMap} calls. If an operation changes elements under
% a smaller context than existing contexts, a separate concrete
% \texttt{LinkedList} is cloned to represent the splitting of partial
% configuration space.

% Although the idea is naive, it can be useful for some application scenarios. For
% example, if a program does not modify list elements very often, a model list
% would contain only a small number of concrete lists. In this case, the overhead
% of transforming \texttt{LinkedList} to handle variational data might be larger
% than the overhead of using \texttt{smap} to invoke methods on a small number of
% concrete \texttt{LinkedList}s, especially for expensive operations like
% searching and sorting. We implemented model classes for \texttt{LinkedList} and
% \texttt{ArrayList} in this naive way.

Depending on different computations in different programs, the effectiveness of
model classes varies. For example, our model \texttt{LinkedList} is optimized
for iterating elements, so programs that iterate lists of optional entries
frequently gain more benefits from our model classes. Our evaluation shows
different levels of improvement after the drop-in replacement of some model
classes, with up to 6 times speedup for GPL.

\paragraph{Custom access patterns}  While custom data structures can store
conditional entries more efficiently, common accessor patterns to iterate over
list entries can still be  very inefficient. For example, getting first the
first, then the second element of a list with optional entries $[1^\alpha,
2^\beta, 3^\gamma, 4, 5]$ would create large conditional values (e.g.,
$\langle \alpha, 1, \langle \beta, 2, \langle \gamma, 3, 4
\rangle\rangle\rangle$ for the first element).

Instead, we detect common access patterns and transform them more
intelligently. Instead of iterating over all elements of the list one by one
(where each element can be a conditional value), we iterate over all optional
elements, where the element is a concrete value, but the iteration is executed
under a restricted variability context based on that element's condition, which
marks under what context this element exists in the list. We integrate the most
recent detection and rewrite of such access patterns of \citet{Lazarek:2017km}.
Our current implementation detects loops that use \texttt{iterator} and
automatically transform such loops to use our specialized list more
efficiently. 

In our evaluation, there were only few instances that benefited from this
optimization, but if they did, the improvements were substantial.  For example,
in \emph{CheckStyle} (see Sec.~\ref{sec:eval}), the program iterates over a
list of $135$ optional checks.  The basic transformation results in exponential
behavior, that makes it infeasible to execute the code without manual rewrites
of the CheckStyle implementation,
whereas our optimization of access patterns allows to execute this code
fragments efficiently.  Overall, several researchers have explored variational
data structures and access patterns recently~\cite{Meng:2017bx,Walkingshaw:2014ch}. 
Model classes and additional
rewrites during the transformation allow us to easily integrate such advances
to improve performance of variational execution on real-world systems.

\subsection{Limitations}

Our current implementation of variational execution has some limitations, most of 
which are related to low-level details of the JVM or restrictions posed by 
the Java runtime (e.g., we cannot directly modify classes in the \texttt{java.lang} package for safety reasons). 
Most limitations are engineering challenges that can be overcome with
additional implementations, typically in the form of model classes.
\looseness=-1

% \begin{itemize}

\paragraph{Exception} We distinguish two types of exceptions: non-variational
exceptions and variational exceptions. \emph{Non-variational exceptions} are thrown or propagated under
the current method context. \emph{Variational exceptions} are thrown or propagated under a smaller
context than the current method context (i.e., only in some partial configurations). 
Semantically, non-variational
exceptions represent cases where invoking a method under \texttt{ctx}
would always result in the same exception under \emph{all configurations} of
\texttt{ctx}, whereas variational exceptions occur \emph{only in some partial configurations}.

Non-variational exceptions are easy to support because the control flow is the
same for all configurations. In fact, we only found non-variational exceptions
when executing our subject systems in evaluation.

Variational exceptions are trickier to handle because method execution might be
interrupted under some partial configurations. If the exception is not caught
inside method invocation, returning from a method results into a normal return
value in some configurations and an uncaught exception in other configurations.  Although it is
possible to support variational exceptions by delaying throwing them
and wrapping them together with normal return values as a conditional value, the
transformation would complicate the control flow of transformed bytecode in a
nontrivial way, especially if exceptions are supposed to be caught inside the
current method or some outer methods. Since variational exceptions are not that
common in our experience, we adopt a less efficient but easier approach to
support variational exceptions: The key idea is to throw an exception
immediately when it occurs and continue the rest of the variational execution
only under the variability context of the exception; then we restart variational 
execution under the remaining contexts that did not result in the previous exception, and keep 
repeating until all contexts have been explored. Re-executions might affect
overall efficiencies of variational execution, but we only observed variational
exceptions in our own artificial examples.

		% Three bytecode instructions are missing in our transformation:
		% exception instruction (\instr{athrow}) and synchronization instructions
		% (\instr{monitorenter} and \instr{monitorexit}).  Although it is possible to
		% support variational exceptions by wrapping them together with results in conditional objects, the transformation process is
		% complicated. Variational exceptions that are only caught in some but
		% not all contexts do not occur in our subject systems, so we leave it as
		% future work.  Instead of actually throwing exceptions, we currently print the
		% exception and its context on the console and continue executing parts
		% that are unaffected, if they exist.  As for synchronization
		% instructions, we currently keep them as is, which implies that we lock
		% sections for all configurations, not just for a variability context;
		% this may lead to over-synchronization and potential liveness issues.

	% \item 
\paragraph{Model classes} We only implemented a handful of model classes (9
classes and in total 1030 lines of Java code) to tackle the environment barrier
required by our subject programs.  We consider all classes that have native
methods and classes that are closely related to internals of the JVM to be
behind the environment barrier and use the strategies discusses above, including
repeated invocations and model classes.  

Currently we support a large set of Java programs, but we may need to provide
more model classes if another program uses certain advanced language features.
We adopt an incremental approach, in which we carefully monitor the need for
model classes at runtime (i.e., when conditional values are passed across
environment barriers). When implementing model classes, our main focus is to
support conditional values. Symbolic execution and model checking face a similar 
challenge, but we argue that the implementation effort is lower for variational
execution, because we compute with concrete values and can therefore
delegate to existing implementations rather than reimplement abstractions of
those operations.

\paragraph{Reflection}
Reflection is relatively simple to handle due to the dynamic nature of our
approach. Since reflection cannot modify bytecode (i.e., we cannot introduce
conditional instructions at runtime), it does not affect our bytecode
transformation of classes. We intercept reflection calls and replace them with
our special call stubs, where we wrap arguments into conditional values, append
variability context to the argument list, and invoke the transformed method,
just as we would transform bytecode statically. We have implemented partial
support for reflection as needed by our subject systems incrementally.

\paragraph{Synchronization} 
Two instructions (\instr{monitorenter} and \instr{monitorexit}) are used to
synchronize concurrent operations. We currently keep them as is, which implies that we lock
sections for all configurations, not just in the current variability context;
this may lead to over-synchronization and potential liveness issues.
		
	% \item 
\paragraph{Array} 
As we will see in Section~\ref{sec:eval}, array operations are
generally expensive in their current form.  Especially when crossing the
environment barrier, we may need to translate conditional arrays into plain
concrete arrays, which can be relatively expensive if the arrays are large. A
more efficient implementation of variational arrays would be future work.

\paragraph{Comparing to VarexJ} Our approach comes with its own limitations, but
most of them can be improved with additional engineering effort. We sidestep most
bottlenecks of the state-of-the-art approach (VarexJ) by transforming bytecode
instead of modifying the underlying JVM. Although the limitations of VarexJ can
potentially also be removed by more engineering effort, we argue that the effort
in VarexJ is much higher because of those additional complexities from the JVM
itself. As an example, native methods are notoriously difficult to support in
Java PathFinder (JPF), the underlying JVM of VarexJ, largely because JPF has its
own memory model for objects, which cannot be passed to native methods directly
and therefore require additional conversion. In contrast, native methods can be
supported in our approach by providing simple model classes to handle
conditional values so that concrete values can be passed to native methods.

% \end{itemize}

\section{Empirical Evaluation}
\label{sec:eval}

In an empirical evaluation, we now execute a number of configurable systems to
assess performance (time and memory consumption) and effectiveness of sharing. 
Specifically, we compare our implementation against
repeatedly executing the unmodified code in all 
configurations (brute-force execution)\footnote{For benchmark programs that have
more than 20 configuration options, we randomly select 1 million valid
configurations for measuring.} and against VarexJ~\cite{MWK+:ASE16}, a
state-of-the-art variational execution engine for Java, which executes bytecode
with a modified virtual machine based on Java Pathfinder.
While performance measures implicitly indicate the benefits of sharing,
we additionally empirically assess how often our VBlock ordering
results in optimal sharing at runtime, especially for methods
for which we cannot guarantee optimal sharing statically.

%  We answer the
% following questions:

% \vspace{0.5em}
% \RQ{1}{How does VarexC perform in terms of execution time?}

% \RQ{2}{Does VarexC consume a lot of memory?}
% \vspace{0.5em}

% Our approach relies on sharing to efficiently explore the entire configuration space.
% RQ1 and RQ2 evaluate the sharing of execution and data, respectively. As
% discussed in Section~\ref{sec:trans} and Section~\ref{sec:control}, our approach
% exploits sharing at the \emph{instruction} level and \emph{VBlock} level. The
% overall performance is largely affected by sharing in instructions, which is
% answered by RQ1. To further explore the efficiency of sharing at VBlock level
% for methods for which we cannot make static guarantees, we answer the question
% below:

% \vspace{0.5em}
% \RQ{3}{How efficient is the sharing of VBlock in practice?}
% \vspace{0.5em}

\subsection{Experimental Setup}

% \begin{table}[!tpb]
% \footnotesize
% \centering
% \begin{tabular}{ l  r  r  r }
% 	\toprule
% 	Program & LOC & \#Options & \#Configurations \\
% 	\midrule
% 	Jetty 7 & 145,421 & 7 & 128 \\
% 	Checkstyle & 14,950 & 141 & $>2^{135}$ \\
% 	Prevayler & 8,975 & 8 & 256 \\
% 	QuEval & 3,109 & 23 & 940 \\
% 	GPL & 662 & 15 & 146 \\
% 	Elevator & 730 & 6 & 20 \\
% 	E-Mail & 644 & 9 & 40 \\
% 	\bottomrule
% \end{tabular}
% \caption{Statistics about benchmark programs used in evaluation. This table
% shows for each program, its size in terms of lines of code, number of (boolean)
% options, and number of valid configurations. 
% We use the same programs and inputs as \cite{Meinicke:2016kc}.}
% \label{tab:bench}
% \end{table}

\paragraph{Benchmarks} Table~\ref{tab:result} shows the benchmark programs used
in this study. For comparability, we use the same set of benchmark programs
from VarexJ~\cite{MWK+:ASE16}, which includes programs from various
domains: Jetty 7 is a HTTP server; Checkstyle is a static coding style checker
for Java programs; Prevayler is an in-memory database system; QuEval is an
academic evaluation framework for database index structures; Elevator, GPL and
E-Mail are commonly used benchmarks from the software product-line community
that are designed to have many variations. These programs have 6 to 141
options, each of which is a \texttt{boolean} controlling inclusion or exclusion
of a feature.  Feature combinations are usually restricted by a feature
model~\citep{SLW:12}. The goal of analyzing these programs is to estimate
the effort of exploring a big configuration space, which can be useful for
testing, static analysis, and so forth. We execute each program with a
representative input, which in each system covers all configuration options and
significant parts of implementation. For example, we feed Checkstyle with 4
Java source code files and use $135$ different checkers to check coding style.
\looseness=-1

\paragraph{Implementation and Hardware} To compare with our tool \emph{VarexC}, we use
the latest VarexJ code base as of 12/12/2017, which includes the most recent
optimizations, added after the last publication. Both VarexC and VarexJ are
executed with Java HotSpot\textsuperscript{TM} 64-bit Server VM (v1.8.0\_161).
We use a laptop with 2.30GHz Intel Core i7 CPU and 16GB system memory. All
results are measured when the machine is idle and unloaded.

\paragraph{Performance Measurement} 
We measure performance in three different settings:
\begin{itemize}
	\item
First, we measure the performance of executing the unmodified program in every single configuration separately on a commodity JVM. Since the execution time may differ significantly between configurations, we report both the 
average execution time (reported as $\mu\text{JVM}$) and the execution time of the slowest
configuration (reported as $\max\text{JVM}$).
	\item Second, we measure the time it takes VarexJ, the state of the art variational interpreter built on top of Java Pathfinder, to execute 
	the program across all configurations (reported as \emph{VarexJ}).
	\item Finally, we measure how long it takes to execute the program across all configurations by executing the modified
	bytecode with a commodity JVM (reported as \emph{VarexC}).
\end{itemize} 
Ideally, the performance of variational execution (\emph{VarexJ} and \emph{VarexC}) would be between
the execution time of the slowest configuration ($\max\text{JVM}$) and the combined 
execution time of all configurations ($\mu\text{JVM}\cdot\text{number of configurations}$): Variational execution needs to at least execute all instructions of
the slowest configuration, but it can usually share effort among multiple configurations.
\looseness=-1

In all three cases, we measure steady-state performance for each
benchmark, based on repeated executions~\cite{Georges:2007br}. Steady-state measurement excludes JVM startup
time, which typically dominates by JIT compilation and class loading. We do not
compare startup performance because VarexJ is implemented as a Java interpreter
itself---in addition to loading classes of benchmark programs, VarexJ needs to
load a lot of necessary classes for the meta-circular interpreter to work, which
would bias our results against VarexJ. For VarexC, we exclude the bytecode
transformation time from measurement because transformation happens once for
each program, similar to compiling source code. We only measure VarexC with
\emph{all optimizations} (see Section~\ref{sec:optim}) for brevity. Following the
suggestion from \citet{Georges:2007br}, we measure steady-state
performance in the following steps:
\begin{enumerate}
	\item Start a JVM invocation $i$ and iterate the benchmark until a
		steady-state is reached, i.e., once the coefficient of variation (CoV)
		of $10$ consecutive iterations falls below a predefined threshold, which
		is $0.02$ in our case.
	\item For the JVM invocation $i$, compute the mean execution time of those
		10 steady iterations, and denote it as $\bar{x}_i$.
	\item Repeat Step (1) and (2) for $10$ times and compute the overall mean
		$\bar{x} = \frac{\sum_{i=1}^{10}\bar{x}_i}{10}$. Finally, we report $\bar{x}$ as
		the measurement result.
\end{enumerate}

In the above measurement, Step (1) and (2) are designed to warm up the JVM,
excluding factors like class loading and JIT compilation. These factors are less
interesting to our evaluation because our main goal is to measure performance of
variational execution. The coefficient of variation threshold is useful for
controlling the effect of garbage collection. Step (3) is designed to minimize
non-determinism of JIT compilation across JVM invocations, because JVM uses
timer-based sampling to drive JIT optimization (e.g., which methods to optimize,
at what level).  Other main sources of non-determinism include thread scheduling
and garbage collection. Thread scheduling is less of a concern for us because
all programs except Jetty are single-threaded. Regarding Jetty, we configure
Jetty to run a small server that has minimal thread scheduling. 
\citet{Georges:2007br} recommends reporting a confidence interval instead of
the mean alone. However, as we will show, the performance difference
between our approach and VarexJ is so large that reporting confidence
intervals is unnecessary. The difference is so obvious and the variation
so small in comparison that statistical tests are not needed.
Due to this large effect size, we omit confidence intervals for
brevity.

\paragraph{Memory usage} To measure memory usage, we calculate the used heap
space by calling APIs of \texttt{java.lang.Runtime} at \emph{every method
entry}, and then record the \emph{maximum} heap space used throughout the
entire JVM invocation.  Even with this frequent sampling, we cannot guarantee
accurate measurement of memory usage, largely because of the non-deterministic
garbage collection and bulk memory allocation. Thus, the memory measurement is
only useful for coarse-grained comparison.  For VarexC and  VarexJ, we perform
each single measurement on a given subject program by executing it \emph{once}.
As a comparison goal, we also measure the memory usage of executing one
representative configuration on a commodity JVM (reported as \emph{JVM}). The
representative configuration is chosen as a valid configuration with the most
features enabled. Since VarexC and VarexJ explore the entire configuration
space, their memory consumption is strictly larger than execution of a single
configuration. To reduce noise, we repeat each measurement 10 times and report
the average.
\looseness=-1

\paragraph{Sharing Efficiency} As discussed in Section~\ref{subsec:properties},
our approach is able to give static guarantees of optimal sharing to methods
that satisfy certain conditions.  To assess sharing for other methods, we
monitor the sharing in our benchmark executions.  Specifically, we collect
traces of which VBlocks are executed under which conditions and subsequently
analyze whether those traces were optimal, with regard to sharing.  For each
variational trace, we expand it into a set of all distinct concrete traces that
it represents, and then compute the alignment of these concrete traces. 
% It would be ideal if we could compute an optimal alignment of these concrete
% traces, but this is challenging because optimal alignment of $n$ traces is
% NP-hard~\cite{Wang:1994jq} and existing algorithms cannot scale to meet our
% needs. 
Since an optimal alignment of $n$ traces is NP-hard~\citep{Wang:1994jq}, 
% Instead of computing an n-way optimal alignment, 
we compute pairwise alignments between all distinct concrete executions using
Needleman-Wunsch algorithm~\cite{Needleman:1970gm}.  If the observed
variational trace is longer than the longest pairwise alignment, we consider
the sharing as not optimal. This pairwise approximation is conservative in that
we may consider executions with optimal sharing as not optimal if the n-way
alignment is longer than the longest pair-wise alignment; conversely, if the
variational trace is not longer than the longest pair-wise alignment we can be
sure that the sharing is optimal. Our pairwise alignment approach sidesteps the
need of computing optimal alignment, but it still has scalability issues if
there are too many pairs, which happen sometimes in our evaluation. For those
cases, we conservatively mark them as suboptimal.  \looseness=-1

\subsection{Execution Time}

\begin{table}[!tpb]
\footnotesize
\centering
\begin{tabular}{l r r r r r r r r r r}
	\toprule
	Subject & LOC & \#Opt & \#Config & $\mu\text{JVM}$ & $\max\text{JVM}$ &
	VarexJ & VarexC & VarexJ/& VarexJ/&VarexC/\\
		   &&&& (in ms) & (in ms) & (in ms) & (in ms) & VarexC & $\max$JVM & $\max$JVM\\
	\midrule
	Jetty & $145,421$ & $7$ & $128$ & $949$ & $1,246$ & $\mathbf{166,340}$ & $\mathbf{4,660}$ & $36$x & $133$x & $4$x \\
	Checkstyle & $14,950$ & $141$ & $>2^{135}$ & $811$ & $946$ & $^{*}\mathbf{89,366}$ & $\mathbf{3,825}$ & $23$x & $94$x & $4$x\\
	Prevayler & $8,975$ & $8$ & $256$ & $13$ & $44$ & $\mathbf{33,124}$ & $\mathbf{725}$ & $46$x & $753$x & $16$x\\
	QuEval & $3,109$ & $23$ & $940$ & $0.03$ & $0.38$ & $2,354$ & $1,244$ & $2$x & $6,195$x & $3,274$x \\
	GPL & $662$ & $15$ & $146$ & $0.55$ & $6.23$ & $4,691$ & $479$ & $10$x & $753$x & $479$x\\
	Elevator & $730$ & $6$ & $20$ & $0.03$ & $0.07$ & $45$ & $7.88$ & $6$x & $643$x & $113$x\\
	E-Mail & $644$ & $9$ & $40$ & $0.02$ & $0.06$ & $21$ & $6.19$ & $3$x & $350$x & $103$x\\
	\bottomrule
\end{tabular}

\caption{\footnotesize 
	Statistics about benchmark programs and performance comparison among JVM, VarexJ and VarexC.
	Statistics include lines of code, number of (\texttt{boolean}) options, and number of valid configurations.
% ``$\mu\text{JVM}$'' and ``$\max\text{JVM}$'' denote the average and maximum
% execution time of all possible configurations (capped at 1 million) using a
% commodity JVM. 
% ``VarexJ'' denotes the execution time of the state-of-the-art
% meta-circular approach. 
% ``VarexC'' denotes the execution time of our bytecode
% transformation approach.  
Numbers in bold denote the cases where VarexC or
VarexJ outperforms brute force execution. 
The last three columns denote the relative speedup or slowdown.
% ``VarexJ/VarexC'' denotes the speedup of
% VarexC over VarexJ. 
% ``VarexJ/$\max$JVM'' and ``VarexC/$\max$JVM'' show the slowdown of VarexJ or VarexC when compared to $\max$JVM.
\\ {\tiny $^{*}$ 
Checkstyle contains a loop over a list of many optional elements. For VarexJ, we had to manually rewrite that loop to allow measurement (due to exponential behavior, we would run out of memory otherwise).}}
\label{tab:result}
\vspace*{-3.0em}
\end{table}

Table~\ref{tab:result} summarizes the performance results, showing
that VarexC outperforms VarexJ by a factor between 2 to 46. 
Variational execution is obviously significantly slower than executing
a single configuration (between 4 and 3200 times slower), but as
configuration spaces grow exponentially, this slowdown is often 
practical to cover the entire space.
\looseness=-1

\paragraph{VarexC vs. VarexJ} Comparing VarexC and VarexJ, we can see that
VarexC outperforms VarexJ in all cases, with a speedup of 2 to 46. To better
understand the speedup, it is useful to divide our subject programs into two
groups and discuss them separately.

QuEval, GPL, Elevator, E-mail are academic examples that only need
basic language features, such as arithmetic computation and array operations.
Thus, a comparison between  VarexC and VarexJ on these programs reveals the
performance gap between bytecode transformation and  interpreter
instrumentation. As we can see, we are up to 10
times faster than  VarexJ, due to lower interpreter overhead and JVM optimizations. 
QuEval is dominated primarily by heavy computations with arrays with only moderate sharing that are expensive in both VarexC and VarexJ.
In a micro-benchmark, we confirmed that sorting on an array of 1000
variational elements with VarexC is roughly 2 times
faster than VarexJ, which likely explains the low performance difference for this program. 
\looseness=-1

Jetty, Prevayler, Checkstyle are medium-sized real-world programs that are widely used
in practice. 
These programs use various more advanced JVM features, including
dynamic class loading (CheckStyle), network access (Jetty) and file access (Prevayler). Since VarexJ is built upon a research JVM,
it inherits limitations from its underlying JVM in this regard,
whereas code transformed with VarexC remains portable across JVMs. 
% In contrast, VarexC does not depend on any specific
% implementation of  JVM. The transformed bytecode generated by VarexC complies
% with the official JVM specification strictly, and thus it is portable to all JVM
% implementations.

\paragraph{VarexC vs. Individual Executions} To investigate how useful
configuration-complete analyses are in practice, we compare VarexC (and for comparison also VarexJ) with the time it takes to execute individual configurations, both average configurations and worst-case configurations.

The overhead of variational execution is generally high, which is explained
both by the instrumentation overhead (creating and propagating conditional
values, boxing, control-flow indirections, SAT solving at runtime), and
by doing the additional work of executing all configurations.
The overhead is usually only justified for large configuration spaces, and
so VarexC (as VarexJ) outperforms the brute-force execution of all configurations only
for Jetty, CheckStyle, and Prevayler.

QuEval, GPL, Elevator, Email represent extreme cases where
variations are used heavily. As we can see from Table~\ref{tab:result}, up to
$940$ configurations are encoded in merely $3,109$ lines of code for QuEval.
When program variations (we called them features interchangeably) present
compactly, the sharing of data and execution becomes less frequent, and thus
explains  why VarexC and VarexJ cannot outperform brute force because
variational execution relies on sharing to be efficient. In fact, there is a
loop in Checkstyle that causes state space explosion for VarexJ because of
looping a list that has $2^{135}$ variants. VarexC uses a model class to handle
this loop gracefully, as discussed in Section~\ref{sec:optim}. However, unlike
these extreme cases, programs in practice often adopt separation of concerns and
thus features do not interact very heavily all the time~\cite{MWK+:ASE16}.

Jetty, Prevayler, Checkstyle implement configuration options such that
they are often orthogonal to each other or have relatively local effects,
which facilitates sharing better. As we can see in Table~\ref{tab:result}, by
exploiting sharing, the performance of VarexC for exploring the entire configuration space is even relatively close to executing only the slowest configuration, with a slowdown as small as a factor of 4.
\looseness=-1

\paragraph{Verdict} We argue that the runtime overhead of VarexC is reasonable
except for one case (QuEval) where expensive array operations with little
sharing dominate the performance.  Runtime overhead does increase for the cases
where interactions of variations are heavily used, but the overhead amortizes
quickly in large configuration spaces, which grow exponentially with the number
of options, unless all options interact. More importantly, research shows that
interactions do not increase with the worst-case exponential behavior in most
cases~\cite{MWK+:ASE16,RSM+:ICSE10}. Even the academic programs that are
designed to interact heavily are still well-behaved with plenty of sharing
despite many interactions. Finally, we argue that the overhead is worthwhile if
we consider the ability to identify all interactions among all options, for which 
the alternative is sampling only a small set of configurations.

% \begin{framed}
%\noindent\RQ{1}{How does VarexC perform in terms of execution time?}\\[0.5em]
%\noindent
\textit{In summary, VarexC outperforms VarexJ with a speedup of 2 to 46 times. The performance gain
comes from various factors, including further optimizations at low level and
portability to mature JVM implementations, all of which benefit from our
strategy of transforming bytecode instead of modifying a language interpreter.
Moreover, VarexC is performant and efficient for practical use in analyzing the
whole configuration space of programs.}
% \end{framed}

\subsection{Memory Usage}

\begin{table}[!tpb]
\footnotesize
\centering
\begin{tabular}{l r r r r r r}
	\toprule
	Subject & JVM &
	VarexJ & VarexC & VarexJ/ & VarexJ/ & VarexC/\\
	 & (in MB) & (in MB) & (in MB) & VarexC & JVM & JVM \\
	\midrule
	Jetty & $268$ & $2,739$ & $648$ & $4.2$ & $10.2$ & $2.4$\\
	Checkstyle & $504$ & $1,106$ & $835$ & $1.3$ & $2.2$ & $1.7$\\
	Prevayler & $65$ & $1,378$ & $288$ & $4.8$ & $21.1$ & $4.4$\\
	QuEval & $59$ & $301$ & $282$ & $1.1$ & $5.1$ & $4.8$\\
	GPL & $141$ & $342$ & $151$ & $2.3$ & $2.4$ & $1.1$\\
	Elevator & $58$ & $92$ & $67$ & $1.4$ & $1.6$ & $1.2$\\
	E-Mail & $59$ & $67$ & $67$ & $1.0$ & $1.1$ & $1.1$\\
	\bottomrule
\end{tabular}

\caption{\footnotesize Memory usage comparison of JVM, VarexJ and VarexC.
% ``JVM'' denotes the memory usage of executing one representative configuration, which is the (valid) one with the most features enabled.
% ``VarexJ'' denotes the state-of-the-art meta-circular approach.
% ``VarexC'' denotes our approach.
}
\vspace*{-2em}
\label{tab:memory}
\end{table}

Table~\ref{tab:memory} summarizes the memory usage results, showing that VarexC
is more memory efficient than VarexJ in all cases
except for a tie in E-Mail. Conceptually, VarexC and VarexJ perform a similar
computation, so the extra memory consumed by VarexJ could result from two main
aspects: less efficient sharing in data and the overhead of the underlying
meta-circular interpreter. As the differences are fairly consistent across
benchmarks, we attribute most efficiency gains to the interpreter's overhead
rather than to differences in sharing. 
Both VarexC and VarexJ, as
expected, consume more memory when compared to the execution of a single
configuration, with the gaps noticeably smaller for VarexC. The memory overhead
of VarexC largely comes from analyzing other configurations. We argue that the
extra memory overhead shown in Table~\ref{tab:memory} is acceptable for modern
machines.

% \begin{framed}
\textit{In summary, VarexC is more memory efficient than VarexJ, due to more efficient sharing in data 
and less overhead from the implementation. Moreover, VarexC has the memory efficiency to 
analyze the entire configuration space in practice.}
% \end{framed}

\subsection{Sharing Efficiency}

\begin{table}[!tpb]
\footnotesize
\centering
% \begin{tabular}{l r r r r r}
% 	\toprule
% 	Subject & Static    & No        & Variational & Not Obviously & Likely \\
% 			& Guarantee & Guarantee & Traces      & Optimal       & Suboptimal \\
% 	\midrule
% 	Jetty & $2,667$ & $257$ & $3,734$ & $0$ & $0$\\
% 	Checkstyle & $2,878$ & $281$ & $268,919$ & $1,199$ & $?$\\
% 	Prevayler & $722$ & $108$ & $5,036$ & $0$ & $0$\\
% 	QuEval & $458$ & $103$ & $9,187$ & $1,233$ & $?$\\
% 	GPL & $244$ & $44$ & $3,476$ & $2$ & $?$\\
% 	Elevator & $119$ & $13$ & $219$ & $75$ & $?$\\
% 	E-Mail & $314$ & $40$ & $120$ & $0$ & $0$\\
% 	\bottomrule
% 	Total	& $7,402$ & $846$ & $290,691$ & $2,509$ &  	\\
% 	\bottomrule
% \end{tabular}

% \caption{
% Sharing efficiency of VarexC. 
% ``Static Guarantee'' denotes the number of methods that we statically guarantee optimality, and ``No Guarantee'' represent other methods.
% ``Variational Traces'' denotes the number of variational traces we collected for those methods with no static guarantees. 
% After filtering and expanding into concrete traces, ``Not Obviously Optimal'' denotes the number of variational traces that require further alignment to confirm optimality.
% ``Likely Suboptimal'' denotes cases where our alignment cannot confirm optimality.
% }
\begin{tabular}{lrrrrr}
	\toprule
			& \multicolumn{2}{c}{Method analysis (static)} & \multicolumn{3}{c}{Method execution (dynamic)} \\ \cmidrule(r){2-3} \cmidrule(l){4-6}
	 		& Guaranteed & No        & Guaranteed & Observed as& Observed as \\
	Subject	& Optimal    & Guarantee & Optimal    & Optimal & Non-Optimal \\
	\midrule
	Jetty & $2,667$ & $257$ & $19,043$ & $3,734$ & $0$\\
	Checkstyle & $2,878$ & $281$ & $1,992,879$ & $268,689$ & $^{*}230$\\
	Prevayler & $722$ & $108$ & $58,274$ & $5,036$ & $0$\\
	QuEval & $458$ & $103$ & $57,383$ & $8,920$ & $267$\\
	GPL & $244$ & $44$ & $34,641$ & $3,476$ & $0$\\
	Elevator & $119$ & $13$ & $2,453$ & $218$ & $1$\\
	E-Mail & $314$ & $40$ & $2,264$ & $120$ & $0$\\
	\bottomrule
	Total	& $7,402$ & $846$ & $2,166,937$ & $290,193$ & $498$	\\
	\bottomrule
\end{tabular}
\caption{\footnotesize Sharing efficiency of VarexC. 
We analyze methods both statically and at runtime. 
At runtime, we distinguish between method executions that are statically guaranteed to be
optimal, that are dynamically observed to be optimal, and that are dynamically observed
to be not optimal.
\\ {\tiny $^{*}$ 
Our alignment analysis has scalability issues with some variational traces of Checkstyle, mainly because 
there are too many features (up to 130) in each single trace, resulting into too
expensive pairwise alignment. 
For those variational traces, we conservatively report them as non-optimal.
}}
\vspace*{-2em}
\label{tab:sharing}
\end{table}

Table~\ref{tab:sharing} shows how efficient our sharing of VBlocks is in
practice. As we can see, we can make static guarantees for $89.7$~percent of all the
methods in our benchmark programs. When observing the executions, those 
methods with static guarantees account for $88.2$~percent of the executed methods,
and we observed that $99.8$ percent for the remaining ones were
optimal as well. The number of method executions that redundantly execute
VBlocks with suboptimal sharing is minimal.

%  For those methods that we cannot guarantee optimal
% sharing, we collect variational traces for further investigation. After
% filtering obvious cases (e.g., all VBlocks were executed under the same context
% at runtime), only $2,369$ of them need further alignment to confirm optimality,
% among which we conservatively declare suboptimal for \{xxxx} cases. In
% summary, VarexC is able to acheive a \{\%} de facto optimality at runtime.

% \begin{framed}
\textit{In summary, sharing in VarexC is efficient, with static guarantees to $89.7\%$ of all
methods. For methods with no static guarantees, VarexC achieves runtime
optimality for $99.8\%$ of those method invocations.}
\looseness=-1
% \end{framed}

\section{Related Work}
\label{related}

We implement variational execution by transforming bytecode.

\paragraph{Variational execution}
Variational execution is a technique to execute a program for
different values while sharing common computations as far as possible.
It has similarities with model checking and symbolic execution,
but performs concrete executions, where multiple concrete values
are distinguished with conditions external to the program,
and focuses on maximizing sharing during the execution by storing
variations in data locally and by aggressively merging control-flow
differences.
Variational execution has a number of existing and potential application scenarios 
in different lines of work. In each case, a program
shall be executed for many similar inputs, typically to observe the
similarities and differences among executions, often with the focus
on interactions among multiple differences.
\begin{itemize}[leftmargin=*]

	\item
	A common use case is testing configurable systems, in which a single
	test case should be executed over a large configuration space.
	For example, \citet{NKN:ICSE14} used variational execution to render
	the content of WordPress while controlling how various plugins interact
	and affect the execution; \citet{MWK+:ASE16} and \citet{KKB:ISSRE12} executed Java programs
	with configuration parameters (as used in our evaluation) to
	observe differences among different configurations.
	Given test cases to provide global or feature-specific specifications,
	variational execution can efficiently check such specifications by
	executing test cases over large configuration spaces~\cite{NKN:ICSE14,Kastner:2012jna,KKB:ISSRE12}.
	\citet{SMS+:vamos18} furthermore used differences among executions as
	clues to find suspicious feature
	interactions.
	\citet{RSM+:ICSE10} used symbolic execution to also detect feature interactions, which however required a lot of effort (80 machine weeks to symbolically execute 319 tests
	with less than 30 configuration options for 10KLOC programs) due to limited sharing abilities of symbolic execution~\cite{MWK+:ASE16}.
	\looseness=-1
	
	\item
	\citet{AF:POPL12} uses variational execution (under the name
	faceted execution) to track information flows in a program.
	In this context, the program is evaluated with sensitive and
	nonsensitive values at the same time, where the equivalent of
	options are decisions who is allowed to see which value.
	In contrast to prior multi-execution work which observes differences
	between two executions, Austin's analysis based on variational
	execution can track interactions among multiple decisions.
	This line of work has been extended with models for variational
	database storage~\cite{YHA+:PLDI16}.
	There are also libraries to enable developers to directly write
	variational programs for this information-flow analysis,
	rather than relying on a variational execution engine~\cite{Schmitz:2016be, AYF+:PLAS13,
schmitzFacetedSecure2018}.

	\item
	Variational execution can further be used to explain the differences
	in program executions among multiple inputs~\cite{2018arXiv}, in line with delta
	debugging~\cite{SZ:ICSE13, Z:FSE02, KKS+:ASPLOS16}.

	\item 
	Variational execution is potentially useful for approaches that
	speculatively change source code or execution to evaluate the consequences.
	For example, mutation testing~\cite{Jia:2011iz} and generate-and-validate 
	automatic program repair~\cite{LeGoues:ci}
	typically try many small changes to the source code and re-execute the
	test suite for each change to evaluate test suite quality or find patches. 
	\citet{Zhang:2007jv} speculatively
	switches predicates in program and re-executes the program to detect
	execution omission errors. \citet{YuriyBrun:2011ve} proactively merges
	different versions and repeatedly executes the test suite to detect
	collaboration conflicts early. By encoding changes as variations,
	variational execution can explore the effects of changes efficiently and uncover
	interesting interactions of changes~\cite{wongBeyond2018}.

	\item
	Finally, variational execution can be used to speed up similar
	computations if there is sufficient sharing to offset the overhead.
	For example, \citet{SBZ+:ICES11} shares similarities among executions
	of simulation workloads and computes with several values in parallel.
	\citet{yingfei:ISSTA18} shares executions of mutated programs with equivalence modulo states 
	in the same process
	and forks new processes only if there are differences in program states
	after executing mutated statements.
	\citet{TXZ09} executes patched and unpatched
	programs together to share redundant computations when testing a patch.
	Variational execution has the potential to scale such use cases to
	exploring interactions among multiple changes.

\end{itemize}

\begin{sloppypar}
Variational execution is fundamentally different from traditional
approaches of multi-execution~\cite{DeGroef:2012jc,DP:SP10, KLZ+:SP12, SAF:SOSP07, HC:ICSE13} and delta debugging~\cite{KKS+:ASPLOS16, Z:FSE02, SZ:ICSE13}
that execute programs repeatedly (either variants of the program
or the same program with different inputs) to compare those executions
to identify, for example, information-flow issues or causes of bugs.
These kinds of approaches execute programs repeatedly in parallel and align those executions
either afterward or through probes at specific points of the executions.
In contrast, variational execution exploits sharing and allows to observe
differences among executions during the execution.
\end{sloppypar}

Ideas similar to variational execution can be found also in approaches
for model checking and symbolic execution~\cite{ALM:SE08,SNG+:FSE,RARF:JPF11},
specifically concepts to store variations as local as possible
to increase sharing and facilitate joining.
Such tools can potentially be used for similar purposes when differences
among inputs are modeled as symbolic decisions, but all other inputs are
concrete. However, as \citet{MWK+:ASE16} has shown, current approaches
are less effective at sharing than the aggressive sharing in variational execution.
\looseness=-1

\paragraph{Implementing Variational Execution}
Existing variational execution approaches (and related approaches) are typically implemented
by modifying the execution engine~\cite{AF:POPL12,MWK+:ASE16,Kastner:2012jna,SNG+:FSE,KKB:ISSRE12, MB:USENIX12, BCD+:FORTE2012}, typically research prototypes or metacircular interpreters
that cause significant overhead and provide only limited support for all
language features. \citet{Schmitz:2016be,
schmitzFacetedSecure2018} provided a library for Haskell with which
users can directly implement programs to use variational execution,
similar to our example in Section~\ref{sec:motivation}.

Instead, we pursue an approach in which we transparently modify Java bytecode
to achieve variational execution on a commodity JVM. Our approach was
inspired by Phosphor~\cite{Bell:2014bt}, a dynamic taint
analysis for Java that tracks taints by instrumenting bytecode.
In contrast to Phosphor, our modifications are significantly more extensive,
as we need not only track additional data, but entirely change how
computations and control flow happen in the program.
CROCHET allows to explore different inputs to the same function by modifying
bytecode to perform checkpoints and rollbacks on the heap of a commodity
JVM~\cite{crochet}. Comparing to CROCHET, our approach can achieve a more
fine-grained sharing of executions while exploring different alternative values.
The only other approach to execute programs variationally with commodity
infrastructure is the implementation behind Jeeves~\cite{YHA+:PLDI16}, that uses metaprogramming
to achieve similar changes for a small subset of Python. Their
transformations are incomplete and not described beyond their implementation
for a small example program.

% In general, (byte)code rewriting is common for dynamic analyses in Java and
% other languages~\cite{examples}, but typically such rewriting is
% significantly less invasive.

\paragraph{Quality assurance for configurable systems}
A main goal of variational execution is testing configurable systems.
There are a wide range of approaches to analyze configurable systems
with large configuration spaces,
typically focused on reusing test cases across product variants, on sampling and on static analysis~\cite{TAK+:CSUR14,MKRGA:ICSE16,NL:CSUR11,PBL05,ER:IST11}.
Sampling strategies analyze or execute a subset of configurations,
but such analysis is neither exhaustive nor does it allow to easily compare
executions~\cite{NL:CSUR11}. For static analyses (including type checking, model checking, and data-flow analysis), researchers have explored many sharing strategies to
encode variability locally (e.g., alternative types for expressions),
to reason about large configuration spaces with propositional formulas,
and to join computations early~\cite{TAK+:CSUR14,LKA+:ESECFSE13}. In a
sense, variational execution can be seen as a generalization of these
sharing techniques for an interpreter~\cite{Kastner:2012jna}.
\citet{BMB+:PLDI13} and \citet{DSB+:STTT17} describe how to lift existing static analyses by providing a variational framework on how
to execute them.

\section{Conclusions}
\label{conclusion}

While variational execution has been applied in different areas such as testing
highly configurable systems and tracking information flow, an efficient
implementation is still missing for practical use. In this work, we propose to
achieve variational execution by transforming programs at the bytecode level.
Our approach is transparent to the developers, and has various advantages such
as making use of underlying optimizations of the JVM and remaining portable to
different JVMs. Our approach transforms individual instructions and modifies
the control flow of methods to exploit sharing of common execution across
configurations. Even with aggressive modification to the control flow decisions,
we formally prove that our transformation to the control flow is correct for
all cases, and optimal for a large subset of cases. We further optimize our
implementation with two different optimizations, each of which optimizes our
approach from different aspects.  With an empirical evaluation on 7 highly
configurable systems, we show that our approach is 2 to 46 times faster while
saving up to 3 quarters of memory usage when compared to the state-of-the-art.
A monitoring at runtime further confirms that we achieve $99.8\%$ optimality
for the methods that we cannot guarantee optimal sharing. Overall, our results
indicate that our approach is useful for analyzing highly configurable systems
in practice.

\begin{acks}

	This work has been supported in part by the NSF (awards
	1318808, 1552944, and 1717022) and AFRL and DARPA (FA8750-16-2-0042).
	Lazarek was supported through Carnegie Mellon's Research Experiences for
	Undergraduates in Software Engineering. We thank Jonathan Bell for his
	advice on bytecode transformation.
	
\end{acks}

%\end{document}  % This is where a 'short' article might terminate

%ACKNOWLEDGMENTS are optional
%
% The following two commands are all you need in the
% initial runs of your .tex file to
% produce the bibliography for the citations in your paper.
\bibliographystyle{ACM-Reference-Format}
\bibliography{lit/MYfull,my,lit/literature}  % sigproc.bib is the name of the Bibliography in this case

\appendix

% \section{Online Extended Version}

% The appendix is available online:
% \url{http://chupanw.github.io/vbc/varexc-extended.pdf} (we plan to move this to arXiv.org after camera ready).

\section{Proofs}

\begin{lemma*}[Disjoint Context Lemma]\label{property:disjoint}
	At any point of execution, the contexts of two different VBlocks are mutually
	exclusive. That is, $\phi(b_i) \land \phi(b_j) = \textit{False}$ for any $i \neq j$.
\end{lemma*}

\begin{proof}

	We prove by induction and case analysis on the jumping targets of a given VBlock.
	In the following, we use $b$ to denote a VBlock, $\phi(b)$ to denote the
	variability context of $b$, and $\phi'(b)$ to denote the new context after context propagation.

	\paragraph{Base case} At the beginning of execution, only the entry VBlock
	has a non-false context. Thus, $\phi(b_i) \land \phi(b_j) = \textit{False}$ because at
	least $\phi(b_i)$ or $\phi(b_j)$ equals $\textit{False}$.

	\paragraph{Induction step} Suppose before execution step $k$, $\phi(b_i)
	\land\phi(b_j)
	= \textit{False}$, for any $i \neq j$. After execution of the next VBlock, say $b_l$, we
	need to update the context of $b_l$'s jumping targets.

	\begin{itemize}
		\item If $b_l$ has only one jumping target $b_m$, according to our context
			propagation, $\phi'(b_l) = \textit{False}$, $\phi'(b_m) = \phi(b_l) \lor \phi(
			b_m)$. Obviously, $\phi'(b_l)$ is
			mutually exclusive to other VBlock context. For any VBlock context,
			say $\phi(b_o)$:

			\begin{equation}
			\begin{split}
				\phi'(b_m) \land \phi(b_o) &= (\phi(b_l) \lor \phi( b_m)) \land
				\phi(b_o) \\
				&= (\phi(b_l) \land \phi( b_o)) \lor (\phi(b_m) \land \phi( b_o))
			\end{split}
			\end{equation}

			According to our induction hypothesis, we have $\phi(b_l) \land
			\phi(b_o) = \textit{False}$
			and $\phi(b_m) \land \phi(b_o) = \textit{False}$, thus induction hypothesis holds after
			execution of $b_l$.

		\item If $b_l$ has two jumping targets $b_m$ and $b_n$, splitting the execution on condition $X$, after executing
			$b_l$, we have $\phi'(b_l) = \textit{False}$, $\phi'(b_m) = \phi(b_m) \lor (X
			\land \phi( b_l))$ and $\phi'(b_n)
			= \phi(b_n) \lor (\neg X \land \phi(b_l))$. For any VBlock context,
			say $\phi(b_o)$:

			\begin{equation}
			\begin{split}
				\phi'(b_m) \land \phi(b_o) &= (\phi(b_m) \lor (X \land
				\phi(b_l))) \land \phi(b_o) \\
				&= (\phi(b_m) \land
				\phi(b_o)) \lor (X \land \phi(b_l) \land \phi( b_o))
			\end{split}
			\end{equation}

			According to our induction hypothesis, we conclude that $\phi'(b_m) \land
			\phi(b_o) = \textit{False}$.
			Similarly, we can conclude $\phi'(b_n) \land \phi(b_o) = \textit{False}$. Moreover:

			\begin{equation}
			\begin{split}
				\phi'(b_m) \land \phi'(b_n) &= (\phi(b_m) \lor (X \land
				\phi(b_l))) \land (\phi(b_n) \lor
				(\neg X \land \phi(b_l)) \\
				&= (\phi(b_m) \land \phi(b_n)) \\
				&\quad \lor (\phi(b_m) \land \neg X \land \phi(b_l)) \\
				&\quad \lor (\phi(b_n) \land X \land \phi(b_l)) \\
				&\quad \lor (X \land \phi(b_l) \land \neg X \land \phi(b_l))
			\end{split}
			\end{equation}

			Again, our induction hypothesis guarantees that $\phi'(b_m)
			\land \phi'(b_n) =
			\textit{False}$. Thus, induction hypothesis holds after execution of $b_l$.

	\end{itemize}
\end{proof}

\begin{property*}[Optimal Sharing Property]

	Given a control flow graph where each node represents a VBlock, our
	variational execution on this graph has optimal sharing if it is acyclic or
	only contains simple loops. A loop is a simple loop if it satisfies the
	following three criteria: (1) has only one loop header; (2) has only one
	exiting node; (3) has no conditional jumps among nodes in the loop.

\end{property*}

As discussed in Section~\ref{subsec:execution}, the actual variational traces
generated by our approach are influenced by the lexical order of VBlocks in the
bytecode.  To help us focus on the essential ideas of proving optimality on
control-flow graphs, we introduce one precondition to the lexical order.

\begin{precond*}
We assume that the strict transitive predecessor relation aligns with the
lexical order of VBlocks in the bytecode. That is, for any pair of VBlocks
$b_i$ and $b_j$, if $b_i$ is a strict transitive predecessor of $b_j$, $b_i$
precedes $b_j$ in the lexical order of bytecode.
\end{precond*}

We also introduce two useful lemmas.

\begin{lemma}\label{lemma:prefix}
	For any two concrete executions of the same simple loop expressed as traces of VBlocks,
	the shorter execution is a prefix of the longer execution.
\end{lemma}

\begin{proof}

	We prove by contradiction. Let us denote the shorter execution as
	$[x_{1}, x_{2}, \ldots, x_{m}]$, and the longer execution as
	$[y_{1}, y_{2}, \ldots, y_{n}]$, where each $x_i$ or $y_j$ represents a VBlock 
	in the control-flow graph and $m \leq n$. Since a simple
	loop has only one loop header and one exiting node, $x_{1}$ must be the
	same as $y_{1}$, and $x_{m}$ must be the same as $y_{n}$.

	For the shorter trace, let us assume it differs from the longer trace at
	the element $x_{i}$ (the $i-th$ element). Thus, $[x_1, x_2, \ldots, x_{i-1}]$ is the 
	same as $[y_1, y_2, \ldots, y_{i-1}]$. Since $x_i$ is different from $y_i$, there must be 
	a conditional jump at $x_{i-1}$ that jumps to either $x_i$ or $y_i$ in the control-flow graph.
	This is contradicting the simple loop criterion that there are no conditional jumps among nodes 
	in the loop.

\end{proof}

\begin{lemma}\label{lemma:vprefix}
	For any variational execution of a simple loop, the variational trace is a prefix of the 
	longest concrete execution trace it represents.
\end{lemma}

\begin{proof}

	We prove by contradiction. Let us denote the variational execution as
	$[v_{1}, v_{2}, \ldots, v_{m}]$, and the longest concrete execution as
	$[x_{1}, x_{2}, \ldots, x_{n}]$, where each $v_i$ or $x_j$ represents a
	VBlock in the control-flow graph. Elements of a variational trace use
	superscripts to indicate variability contexts of execution, but they are
	less important in this proof so we omit them for brevity. Since a simple
	loop has only one loop header and one exiting node, $v_{1}$ must be the
	same as $x_{1}$, and $v_{m}$ must be the same as $x_{n}$.

	Let us assume the variational trace differs from the longest trace at the
	element $v_{i}$ (the $i-th$ element). Thus, $[v_1, v_2, \ldots, v_{i-1}]$
	is the same as $[x_1, x_2, \ldots, x_{i-1}]$. Since $v_i$ is different from
	$x_i$, there could be two causes. First, there is a conditional jump at
	$x_{i-1}$ that jumps to either $v_i$ or $x_i$ in the control-flow graph.
	Second, during variational execution of the loop, two different VBlocks have
	satisfiable contexts, which also requires at least a conditional jump among
	VBlocks in the loop because conditional jumps are the only places where we
	split variability contexts.  Both of these cases contradict the simple loop
	criterion that there are no conditional jumps among nodes in the loop.

\end{proof}

% \begin{lemma}\label{}
% \end{lemma}

With the precondition and lemmas above, we will prove the original property
below. Again, we prove that, given a control flow graph of VBlocks, our
variational execution on this graph has optimal sharing if it is acyclic or
only contains simple loops.

\begin{proof}

	We prove by case analysis on acyclic control-flow graphs and control-flow
	graphs with simple loops, respectively.

	\paragraph{Acyclic} For any acyclic control-flow graph, suppose our variational execution generates a trace $t_v$
	with $n$ elements. Our static partial ordering between VBlocks ensures that
	these $n$ elements are different. Otherwise, suppose $b_i$ appears twice in
	$t_v$, there must be a transitive predecessor of $b_i$ between these two appearances of
	$b_i$ in $t_v$ because the control-flow graph is acyclic. However, this is impossible
	because $b_i$'s transitive predecessors can only precede $b_i$ in our
	variational traces, due to our static partial ordering. 
	
	As discussed in Section~\ref{subsec:properties}, $t_v$ represents a set of concrete execution traces under
	different restricted contexts. These concrete traces have the following two
	properties:
	
	\begin{itemize}

		\item There is no duplicated VBlock in each concrete trace, because the
			control-flow graph is acyclic.
			
		\item The $n$ different VBlocks in $t_v$ must appear in one or more of these
			concrete traces, because our variational execution only executes
			VBlocks with satisfiable contexts.  
	\end{itemize}
	
	We denote the optimal sharing of these concrete traces as $t_o$. From these
	two properties, we know that $length(t_o) = n$ because each VBlock must
	occur at least once, and at most once if the traces are optimally aligned.
	So, the length of the optimal alignment must be $n$.  Since $length(t_v)$ is
	also $n$, we achieve optimal sharing.

	\paragraph{Simple Loop} For any control-flow graph with one or more simple
	loops, we denote a loop as $L_i$, with the subscript distinguishing
	different loops. Suppose our variational execution generates a trace $t_v$.
	Our static partial ordering guarantees that $t_v$ has the following
	properties:

	\begin{itemize}

		\item If a loop $L_i$ is executed, VBlocks belonging to $L_i$
			are adjacent to each other in $t_v$, without any VBlock that does
			not belong to $L_i$ in between. We call this region a looping
			region of $L_i$, denoted as $RV_i$. This can be proven by
			contradiction. If there is a VBlock $b$ (not belonging to $L_i$)
			inside $RV_i$, between $b_x$ and $b_y$ (both $b_x$ and $b_y$ are
			part of the loop $L_i$), $b$ must have the same transitive
			predecessor relation with $b_x$ and $b_y$, because $b$ is not part
			of the loop $L_i$. If this is the case, our static partial ordering
			would require $b$ to either precede both $b_x$ and $b_y$ or fall
			behind $b_x$ and $b_y$ in the trace. This is contracting to the
			assumption that $b$ is between $b_x$ and $b_y$ in the trace $t_v$.

		\item In $t_v$, any VBlock $b$ that is outside looping regions have no
			duplication. This can also be proven by contradiction. Suppose $b$
			(not belonging to any loops) appears twice in $t_v$, there must be
			a transitive predecessor of $b$ between these two appearances of
			$b$ in $t_v$ because $b$ does not belong to any loops. However,
			$b$'s transitive predecessor cannot appear between two occurrences
			of $b$ in $t_v$, due to our static partial ordering.

		\item For any loop $L_i$, there is at most one looping region $RV_i$ in
			$t_v$. Otherwise, $L_i$ must be a inner loop of another bigger
			loop. If this is the case, there must be at least one conditional
			jump in the outer loop, and therefore the outer loop fails to
			satisfy the simple loop premise.

	\end{itemize}

	Based on these properties, we have $length(t_v) =
	\sum\limits_{i}^{}length(RV_i) + n$, where $length(RV_i)$ denotes the number
	of elements in $RV_i$, and $n$ denotes the number of VBlocks in $t_v$ that
	are not part of any loops.

	Now if we consider the concrete traces represented by $t_v$, in order to
	produce the optimal merging of these traces $t_o$, we need to take two
	steps: (1) merge looping regions of concrete traces and (2) merge VBlocks
	that do not belong to any loop.

	\begin{enumerate}

		\item From Lemma~\ref{lemma:prefix}, we know that the length of merging all
			looping regions of $L_i$ across concrete traces is determined by
			the longest looping region, which we denote as $length(RMax_i)$.

		\item Merging VBlocks that do not belong to any loop would result in $n$
			elements. This is equivalent to merging concrete traces of acyclic
			control-flow graphs, which we have already proven in the first half of this proof.

	\end{enumerate}

	Thus, we have $length(t_o) = \sum\limits_{i}^{}length(RMax_i) + n$. For any
	loop $L_i$, the length of its looping region $RV_i$ in $t_v$ (if exists) is
	bounded by $length(RMax_i)$, (i.e., $length(RV_i) \leq length(RMax_i)$), as
	we have proven in Lemma~\ref{lemma:vprefix}. On the other hand, $RV_i$ is
	guaranteed to represent the longest looping of $L_i$ in concrete traces
	because the context with the longest looping must be executed to satisfy
	correctness. So, $length(RV_i) \geq length(RMax_i)$, which gives us
	$length(RV_i) = length(RMax_i)$. Since $length(t_v) = length(t_o)$, we
	achieve optimal sharing.

\end{proof}

\section{Example of Transformed Bytecode}
\input{MotivationExampleBytecode}

\end{document}

%% file: macros.tex
\newcommand{\code}[1]{\texttt{\footnotesize #1}}
\newcommand{\todo}[1]{{\color{magenta} \bf \{TODO: {#1}\}}}

\def\modify#1#2#3{{\small\underline{\sf{#1}}:} {\color{red}{\small #2}}
{{\color{red}\mbox{$\Rightarrow$}}} {\color{blue}{#3}}}

\newcommand{\cpwmodify}[2]{\modify{Chu-Pan}{#1}{#2}}

\newcommand\mymargin[1]{\marginpar{{\flushleft\textsc\footnotesize {#1}}}}
\newcommand\cpwmargin[1]{\mymargin{CPW:\;#1}}

\newcommand{\figref}[1]{Figure~\ref{#1}}
\newcommand{\eqnref}[1]{violation~\eqref{#1}}
\newcommand{\secref}[1]{Section~\ref{#1}}
\newcommand{\tblref}[1]{Table~\ref{#1}} 

\newcommand{\cpwmodifyok}[2]{#2}
\newcommand{\cpwmodifyno}[2]{#1} 
\newcommand{\cpwrevise}[2]{\modify{Chu-Pan}{#1}{#2}}

\newcommand{\note}[2]{{\small\color{magenta}\underline{\sf{#1}}:} {\color{magenta}{\small #2}}}
\newcommand{\cpwnote}[1]{\note{Chu-Pan}{#1}}

\renewcommand{\cpwrevise}[2]{#2}

\renewcommand{\cpwmodify}[2]{#2}

% Customized commands
\renewcommand{\paragraph}[1]{\vspace{.4em}\noindent\textbf{#1.} } 
\newcommand{\instr}[1]{\texttt{#1}}
\newcommand{\JVM}{Java Virtual Machine}

%% file: MotivationExampleBytecode.tex
As a concrete example of our bytecode transformation, we show the original and
transformed bytecode of the \texttt{getWeather()} method in
Figure~\ref{fig:running}, which covers multiple kinds of bytecode instructions.

\begin{lstlisting}[language = JVMIS, 
					tabsize = 2,
					escapechar = ~, 
					basicstyle = \footnotesize\ttfamily\bfseries, 
					breaklines = true
					mathescape = true,
					numbers = none,
					title = Original bytecode for the getWeather() method shown in Figure~1,
					name = unlifted,
					label = {lst:unlifted},
					frame            = tb,    % draw frame at top and bottom of code block
					framesep         = 3pt,   % expand outward
					framerule        = 0.4pt, % expand outward 
					commentstyle     = \color{Green},      % comment color
					keywordstyle     = \color{blue},       % keyword color
					% stringstyle      = \color{DarkRed},    % string color
					% backgroundcolor  = \color{WhiteSmoke}, % backgroundcolor color
					showstringspaces = false              % do not mark spaces in strings
					]
  public static java.lang.String getWeather();
    Code:
       0: invokestatic  #19                 // Method getCelsius:()F
       3: fstore_0
       4: getstatic     #20                 // Field FAHRENHEIT:Z
       7: ifeq          36
      10: new           #15                 // class java/lang/StringBuilder
      13: dup
      14: invokespecial #16                 // Method java/lang/StringBuilder."<init>":()V
      17: fload_0
      18: ldc           #21                 // float 1.8f
      20: fmul
      21: ldc           #22                 // float 32.0f
      23: fadd
      24: invokevirtual #23                 // Method java/lang/StringBuilder.append:(F)Ljava/lang/StringBuilder;
      27: ldc           #24                 // String F
      29: invokevirtual #17                 // Method java/lang/StringBuilder.append:(Ljava/lang/String;)Ljava/lang/StringBuilder;
      32: invokevirtual #18                 // Method java/lang/StringBuilder.toString:()Ljava/lang/String;
      35: areturn
      36: new           #15                 // class java/lang/StringBuilder
      39: dup
      40: invokespecial #16                 // Method java/lang/StringBuilder."<init>":()V
      43: fload_0
      44: invokevirtual #23                 // Method java/lang/StringBuilder.append:(F)Ljava/lang/StringBuilder;
      47: ldc           #25                 // String C
      49: invokevirtual #17                 // Method java/lang/StringBuilder.append:(Ljava/lang/String;)Ljava/lang/StringBuilder;
      52: invokevirtual #18                 // Method java/lang/StringBuilder.toString:()Ljava/lang/String;
      55: areturn
\end{lstlisting}

\begin{lstlisting}[language = JVMIS, 
					tabsize = 2,
					escapechar = ~, 
					basicstyle = \footnotesize\ttfamily\bfseries, 
					breaklines = true
					mathescape = true,
					numbers = none,
					title = Transformed bytecode for the getWeather() method shown in Figure~1,
					name = lifted,
					label = {lst:lifted},
					frame            = tb,    % draw frame at top and bottom of code block
					framesep         = 3pt,   % expand outward
					framerule        = 0.4pt, % expand outward 
					commentstyle     = \color{Green},      % comment color
					keywordstyle     = \color{blue},       % keyword color
					% stringstyle      = \color{DarkRed},    % string color
					% backgroundcolor  = \color{WhiteSmoke}, % backgroundcolor color
					showstringspaces = false              % do not mark spaces in strings
					]
  public static edu.cmu.cs.varex.V<java.lang.String> getWeather____Ljava_lang_String(de.fosd.typechef.featureexpr.FeatureExpr);
    Code:
       0: invokestatic  #18                 // Method de/fosd/typechef/featureexpr/FeatureExprFactory.False:()Lde/fosd/typechef/featureexpr/FeatureExpr;
       3: astore        7
       5: invokestatic  #18                 // Method de/fosd/typechef/featureexpr/FeatureExprFactory.False:()Lde/fosd/typechef/featureexpr/FeatureExpr;
       8: astore        6
      10: invokestatic  #18                 // Method de/fosd/typechef/featureexpr/FeatureExprFactory.False:()Lde/fosd/typechef/featureexpr/FeatureExpr;
      13: astore        5
      15: invokestatic  #18                 // Method de/fosd/typechef/featureexpr/FeatureExprFactory.False:()Lde/fosd/typechef/featureexpr/FeatureExpr;
      18: astore        4
      20: invokestatic  #24                 // Method edu/cmu/cs/varex/One.getOneNull:()Ledu/cmu/cs/varex/V;
      23: astore_3
      24: invokestatic  #24                 // Method edu/cmu/cs/varex/One.getOneNull:()Ledu/cmu/cs/varex/V;
      27: astore_1
      28: invokestatic  #24                 // Method edu/cmu/cs/varex/One.getOneNull:()Ledu/cmu/cs/varex/V;
      31: astore_2
      32: aload_0
      33: astore        4
      35: aload         4
      37: invokestatic  #180                // Method getCelsius____F:(Lde/fosd/typechef/featureexpr/FeatureExpr;)Ledu/cmu/cs/varex/V;
      40: aload         4
      42: swap
      43: aload_1
      44: invokestatic  #59                 // InterfaceMethod edu/cmu/cs/varex/V.choice:(Lde/fosd/typechef/featureexpr/FeatureExpr;Ledu/cmu/cs/varex/V;Ledu/cmu/cs/varex/V;)Ledu/cmu/cs/varex/V;
      47: astore_1
      48: getstatic     #182                // Field FAHRENHEIT:Ledu/cmu/cs/varex/V;
      51: invokestatic  #70                 // Method edu/cmu/cs/varex/VOps.whenEQ:(Ledu/cmu/cs/varex/V;)Lde/fosd/typechef/featureexpr/FeatureExpr;
      54: dup
      55: invokestatic  #76                 // Method edu/cmu/cs/varex/VCache.not:(Lde/fosd/typechef/featureexpr/FeatureExpr;)Lde/fosd/typechef/featureexpr/FeatureExpr;
      58: aload         4
      60: invokestatic  #80                 // Method edu/cmu/cs/varex/VCache.and:(Lde/fosd/typechef/featureexpr/FeatureExpr;Lde/fosd/typechef/featureexpr/FeatureExpr;)Lde/fosd/typechef/featureexpr/FeatureExpr;
      63: aload         5
      65: invokeinterface #85,  2           // InterfaceMethod de/fosd/typechef/featureexpr/FeatureExpr.or:(Lde/fosd/typechef/featureexpr/FeatureExpr;)Lde/fosd/typechef/featureexpr/FeatureExpr;
      70: astore        5
      72: aload         4
      74: invokestatic  #80                 // Method edu/cmu/cs/varex/VCache.and:(Lde/fosd/typechef/featureexpr/FeatureExpr;Lde/fosd/typechef/featureexpr/FeatureExpr;)Lde/fosd/typechef/featureexpr/FeatureExpr;
      77: aload         6
      79: invokeinterface #85,  2           // InterfaceMethod de/fosd/typechef/featureexpr/FeatureExpr.or:(Lde/fosd/typechef/featureexpr/FeatureExpr;)Lde/fosd/typechef/featureexpr/FeatureExpr;
      84: astore        6
      86: invokestatic  #18                 // Method de/fosd/typechef/featureexpr/FeatureExprFactory.False:()Lde/fosd/typechef/featureexpr/FeatureExpr;
      89: astore        4
      91: aload         5
      93: invokestatic  #89                 // Method edu/cmu/cs/varex/VCache.isContradiction:(Lde/fosd/typechef/featureexpr/FeatureExpr;)Z
      96: ifne          205
      99: new           #138                // class model/java/lang/StringBuilder
     102: dup
     103: aload         5
     105: invokespecial #140                // Method model/java/lang/StringBuilder."<init>":(Lde/fosd/typechef/featureexpr/FeatureExpr;)V
     108: aload_1
     109: ldc           #183                // float 1.8f
     111: invokestatic  #189                // Method java/lang/Float.valueOf:(F)Ljava/lang/Float;
     114: aload         5
     116: swap
     117: invokestatic  #95                 // InterfaceMethod edu/cmu/cs/varex/V.one:(Lde/fosd/typechef/featureexpr/FeatureExpr;Ljava/lang/Object;)Ledu/cmu/cs/varex/V;
     120: aload         5
     122: invokestatic  #193                // Method edu/cmu/cs/varex/VOps.fmul:(Ledu/cmu/cs/varex/V;Ledu/cmu/cs/varex/V;Lde/fosd/typechef/featureexpr/FeatureExpr;)Ledu/cmu/cs/varex/V;
     125: ldc           #194                // float 32.0f
     127: invokestatic  #189                // Method java/lang/Float.valueOf:(F)Ljava/lang/Float;
     130: aload         5
     132: swap
     133: invokestatic  #95                 // InterfaceMethod edu/cmu/cs/varex/V.one:(Lde/fosd/typechef/featureexpr/FeatureExpr;Ljava/lang/Object;)Ledu/cmu/cs/varex/V;
     136: aload         5
     138: invokestatic  #197                // Method edu/cmu/cs/varex/VOps.fadd:(Ledu/cmu/cs/varex/V;Ledu/cmu/cs/varex/V;Lde/fosd/typechef/featureexpr/FeatureExpr;)Ledu/cmu/cs/varex/V;
     141: aload         5
     143: invokevirtual #200                // Method model/java/lang/StringBuilder.append__F__Lmodel_java_lang_StringBuilder:(Ledu/cmu/cs/varex/V;Lde/fosd/typechef/featureexpr/FeatureExpr;)Ledu/cmu/cs/varex/V;
     146: ldc           #202                // String F
     148: aload         5
     150: swap
     151: invokestatic  #95                 // InterfaceMethod edu/cmu/cs/varex/V.one:(Lde/fosd/typechef/featureexpr/FeatureExpr;Ljava/lang/Object;)Ledu/cmu/cs/varex/V;
     154: invokedynamic #207,  0            // InvokeDynamic #5:apply:(Ledu/cmu/cs/varex/V;)Ljava/util/function/BiFunction;
     159: aload         5
     161: invokeinterface #121,  3          // InterfaceMethod edu/cmu/cs/varex/V.sflatMap:(Ljava/util/function/BiFunction;Lde/fosd/typechef/featureexpr/FeatureExpr;)Ledu/cmu/cs/varex/V;
     166: invokedynamic #212,  0            // InvokeDynamic #6:apply:()Ljava/util/function/BiFunction;
     171: aload         5
     173: invokeinterface #121,  3          // InterfaceMethod edu/cmu/cs/varex/V.sflatMap:(Ljava/util/function/BiFunction;Lde/fosd/typechef/featureexpr/FeatureExpr;)Ledu/cmu/cs/varex/V;
     178: aload         5
     180: swap
     181: aload_2
     182: invokestatic  #59                 // InterfaceMethod edu/cmu/cs/varex/V.choice:(Lde/fosd/typechef/featureexpr/FeatureExpr;Ledu/cmu/cs/varex/V;Ledu/cmu/cs/varex/V;)Ledu/cmu/cs/varex/V;
     185: astore_2
     186: aload         5
     188: aload         7
     190: invokeinterface #85,  2           // InterfaceMethod de/fosd/typechef/featureexpr/FeatureExpr.or:(Lde/fosd/typechef/featureexpr/FeatureExpr;)Lde/fosd/typechef/featureexpr/FeatureExpr;
     195: astore        7
     197: invokestatic  #18                 // Method de/fosd/typechef/featureexpr/FeatureExprFactory.False:()Lde/fosd/typechef/featureexpr/FeatureExpr;
     200: astore        5
     202: goto          205
     205: aload         6
     207: invokestatic  #89                 // Method edu/cmu/cs/varex/VCache.isContradiction:(Lde/fosd/typechef/featureexpr/FeatureExpr;)Z
     210: ifne          287
     213: new           #138                // class model/java/lang/StringBuilder
     216: dup
     217: aload         6
     219: invokespecial #140                // Method model/java/lang/StringBuilder."<init>":(Lde/fosd/typechef/featureexpr/FeatureExpr;)V
     222: aload_1
     223: aload         6
     225: invokevirtual #200                // Method model/java/lang/StringBuilder.append__F__Lmodel_java_lang_StringBuilder:(Ledu/cmu/cs/varex/V;Lde/fosd/typechef/featureexpr/FeatureExpr;)Ledu/cmu/cs/varex/V;
     228: ldc           #214                // String C
     230: aload         6
     232: swap
     233: invokestatic  #95                 // InterfaceMethod edu/cmu/cs/varex/V.one:(Lde/fosd/typechef/featureexpr/FeatureExpr;Ljava/lang/Object;)Ledu/cmu/cs/varex/V;
     236: invokedynamic #219,  0            // InvokeDynamic #7:apply:(Ledu/cmu/cs/varex/V;)Ljava/util/function/BiFunction;
     241: aload         6
     243: invokeinterface #121,  3          // InterfaceMethod edu/cmu/cs/varex/V.sflatMap:(Ljava/util/function/BiFunction;Lde/fosd/typechef/featureexpr/FeatureExpr;)Ledu/cmu/cs/varex/V;
     248: invokedynamic #224,  0            // InvokeDynamic #8:apply:()Ljava/util/function/BiFunction;
     253: aload         6
     255: invokeinterface #121,  3          // InterfaceMethod edu/cmu/cs/varex/V.sflatMap:(Ljava/util/function/BiFunction;Lde/fosd/typechef/featureexpr/FeatureExpr;)Ledu/cmu/cs/varex/V;
     260: aload         6
     262: swap
     263: aload_2
     264: invokestatic  #59                 // InterfaceMethod edu/cmu/cs/varex/V.choice:(Lde/fosd/typechef/featureexpr/FeatureExpr;Ledu/cmu/cs/varex/V;Ledu/cmu/cs/varex/V;)Ledu/cmu/cs/varex/V;
     267: astore_2
     268: aload         6
     270: aload         7
     272: invokeinterface #85,  2           // InterfaceMethod de/fosd/typechef/featureexpr/FeatureExpr.or:(Lde/fosd/typechef/featureexpr/FeatureExpr;)Lde/fosd/typechef/featureexpr/FeatureExpr;
     277: astore        7
     279: invokestatic  #18                 // Method de/fosd/typechef/featureexpr/FeatureExprFactory.False:()Lde/fosd/typechef/featureexpr/FeatureExpr;
     282: astore        6
     284: goto          287
     287: aload_2
     288: aload         7
     290: invokestatic  #168                // Method edu/cmu/cs/varex/VOps.verifyAndThrowException:(Ledu/cmu/cs/varex/V;Lde/fosd/typechef/featureexpr/FeatureExpr;)Ledu/cmu/cs/varex/V;
     293: areturn
\end{lstlisting}